\documentclass[11pt,a4paper]{article}
\usepackage{graphicx,amsfonts,amssymb,mathrsfs}
\usepackage{amsopn}
\usepackage{amsmath,amsthm}

\DeclareMathOperator{\C}{\mathbb{C}}

\DeclareMathOperator{\N}{\mathbb{N}}

\DeclareMathOperator{\re}{Re}

\DeclareMathOperator{\R}{\mathbb{R}}

\DeclareMathOperator{\sgn}{\textup{sgn}}
\DeclareMathOperator{\supp}{\textup{supp}}

\DeclareMathOperator{\Z}{\mathbb{Z}}
\DeclareMathAlphabet{\mathpzc}{OT1}{pzc}{m}{it}

\newcommand{\bDelta}{{\mbox{$\triangle$}\hspace{-8.0pt}\scalebox{0.8}{$\triangle$}}}

\newcommand{\bA}{\boldsymbol{A}}

\newcommand{\bk}{\boldsymbol{k}}
\newcommand{\bR}{\boldsymbol{R}}
\newcommand{\brr}{\boldsymbol{r}}
\newcommand{\bx}{\boldsymbol{x}}

\newcommand{\by}{\boldsymbol{y}}

\newcommand{\brho}{\boldsymbol{\rho}}
\newcommand{\sx}{\mathrm{x}}
\newcommand{\sR}{\mathrm{R}}

\newcommand{\uA}{\textup{A}}

\newcommand{\uH}{\textup{CS}} 
\newcommand{\uS}{\textup{LL}} 
\newcommand{\Dc}{\mathscr{D}}
\newcommand{\Ec}{\mathcal{E}}

\newcommand{\Pc}{\mathcal{P}}

\newcommand{\ep}{\varepsilon}
\newcommand{\bep}{{\boldsymbol{\varepsilon}}}

\theoremstyle{plain}
\newtheorem{thm}{Theorem}

\newtheorem{lem}[thm]{Lemma}
\newtheorem*{cor}{Corollary}
\newtheorem{prop}[thm]{Proposition}

\theoremstyle{definition}

\theoremstyle{remark}

\newtheorem*{ack}{Acknowledgements}

\begin{document}

\title{Local exclusion and Lieb-Thirring inequalities for intermediate and fractional statistics}

%
%
%
%

\author{Douglas Lundholm$^{a,}$\thanks{Work partly done while visiting FIM, ETH Z\"urich. (e-mail: lundholm@math.ku.dk)} 
\ and Jan Philip Solovej$^{b,}$\thanks{Work partly done while visiting Institut Mittag-Leffler. (e-mail: solovej@math.ku.dk)}
\\[5pt]
\scriptsize $^a$ Institut Mittag-Leffler, Aurav\"agen 17, SE-182 60 Djursholm, Sweden\\ 
\scriptsize $^b$ Department of Mathematical Sciences, University of Copenhagen\\ 
\scriptsize Universitetsparken 5, DK-2100 Copenhagen \O, Denmark
}



\date{}

\maketitle

\begin{abstract}
	In one and two spatial dimensions there is a logical possibility
	for identical quantum particles different from bosons and fermions,
	obeying intermediate or fractional (anyon) statistics.
	We consider applications of a recent Lieb-Thirring inequality 
	for anyons in two dimensions,
	and derive new Lieb-Thirring inequalities for intermediate
	statistics in one dimension with implications for models of 
	Lieb-Liniger and Calogero-Sutherland type.
	These inequalities follow from a local form of the exclusion principle
	valid for such generalized exchange statistics.


\end{abstract}



\section{Introduction}

	Fundamental to 
	many-body quantum mechanics is the notion of
	identical particles and associated particle statistics.
	Many of the remarkable quantum phenomena observed in nature
	are direct consequences of the fact that identical particles
	in three spatial dimensions are 
	either \emph{bosons}, obeying Bose-Einstein statistics,
	or \emph{fermions}, obeying Fermi-Dirac statistics.
	For example, the amplification of light in lasers
	by superposition of photons is possible because they are bosons, 
	while the periodic table of elements arises from the sequential 
	filling of atomic shells by electrons which are fermions.
	This binary classification has been known since the early days of quantum
	theory, and follows logically from the requirement 
	that an $N$-body wave function
	$\psi(\bx_1,\ldots,\bx_N)$, 
	modeled as a square-integrable complex-valued function\footnote{
	We will for simplicity always model particles as non-relativistic 
	and \emph{scalar}, i.e. without internal symmetries or spin.}
	on the $N$-particle configuration space $(\R^3)^N$, 
	or alternatively as an element of the tensor product 
	of one-particle spaces
	$\bigotimes^N L^2(\R^3)$,
	has to be either symmetric (corresponding to bosons) 
	or antisymmetric (corresponding to fermions) under the action
	of permutations of the particle labels 
	--- which, 
	if the particles are 
	indistinguishable,
	must leave the observable probability $|\psi|^2$ invariant.
	However, it was discovered in the 
	1970-80's \cite{Streater-Wilde:70,Leinaas-Myrheim:77,Goldin-Menikoff-Sharp:81,Wilczek:82} 
	that this simple picture 
	for identical particles and statistics
	is logically incomplete 
	and should be replaced by a more general framework, 
	which in three (or higher) space dimensions nicely reduces to 
	the above binary picture,
	but in the one and two dimensional cases actually admits 
	a wider range of possibilities.
	Identical particles classified via this generalized framework,
	but falling outside the usual boson/fermion classification, 
	were said to obey \emph{intermediate} or \emph{fractional statistics}.
	Although \emph{elementary} identical particles live in 
	three space dimensions and hence must be either bosons or fermions, 
	these more exotic one- and two-dimensional possibilities 
	have recently become more than just 
	a mathematical curiosity,
	with the advent of trapped bosonic condensates 
	\cite{Kinoshita-Wenger-Weiss:04,Paredes_et_al:04}
	and quantum Hall physics 
	\cite{Arovas-Schrieffer-Wilczek:84,Laughlin:99,Prange-Girvin:90},
	and thereby the discoveries of 
	effective models of particles (or \emph{quasi}-particles)
	which seem to obey these generalized rules for identical 
	particles and statistics
	(see \cite{Bloch-Dalibard-Zwerger:08,Froehlich:90,Khare:05,Lerda:92,Myrheim:99,Ouvry:07,Polychronakos:99,Wilczek:90} 
	for reviews).

	In two space dimensions the usual requirement, 
	that the phase change of the
	wave function which arises from an interchange of two identical particles
	(or, in the more precise general framework, 
	from a \emph{continuous} simple interchange of two particles)
	needs to be either +1 (boson) or -1 (fermion), is no longer valid.
	This interchange phase is instead allowed to be \emph{any} 
	fixed unit complex number $e^{i\alpha\pi}$.
	The corresponding particles are called \emph{anyons},
	and are alternatively parameterized by the real-valued 
	statistics parameter $\alpha \in (-1,1]$.
	A standard way to model such anyons 
	is by means of bosons in $\R^2$ together with magnetic potentials 
	$\bA_j(\sx)$ 
	of Aharonov-Bohm-type with strength $\alpha$ 
	between every pair of particles
	(see \eqref{kinetic_energy_A} below),
	giving rise to the correct phase of the wave function
	as the particles interchange or encircle each other
	\cite{Leinaas-Myrheim:77,Wu:84}.
	The free kinetic energy operator for $N$ anyons 
	is thus given by
	(see \eqref{kinetic_energy_A}; we use the unit conventions $\hbar = m = 1$)
	$$
		\hat{T}_{\uA} := \frac{1}{2} \sum_{j=1}^N D_j^2,
		\qquad D_j := -i\nabla_j + \bA_j(\sx),
	$$
	and acting on completely symmetric wave functions
	$\psi \in \bigotimes^N_{\textup{sym}} L^2(\R^2)$.
	
	In one dimension, the above geometric picture of 
	continuous inter\-chan\-ges of particles
	breaks down 
	--- particles have to collide in order to be interchanged.
	In a quantum mechanical context this necessitates 
	the prescription of boundary conditions 
	for the wave function at the collision points
	$r := x_{j+1}-x_j \to 0^+$
	(i.e. at the boundary of the proper configuration space
	$X := \{ \sx \in \R^N : x_1 < x_2 < \ldots < x_N \}$,
	in this case with codimension one).
	There are two ways in which one may arrive naturally at a set of
	boundary conditions. 
	One approach \cite{Leinaas-Myrheim:77,Aneziris_et_al:91}, 
	which is generally referred to as \emph{Schr\"odinger-type quantization}, 
	starts out with the proper configuration space $X$
	and imposes a general (Robin) boundary condition of the form
	$\partial_r \psi = \eta \psi$ at $r=0^+$,
	with $\eta \in \R$ an arbitrary but fixed statistics parameter.
	The case $\eta=0$ corresponds to bosons (Neumann b.c.)
	and $\eta = +\infty$ to fermions (Dirichlet b.c.).
	Such boundary conditions can, 
	if $\psi$ is extended from $X$ to $\R^N$ and modeled symmetric about $r=0$,
	alternatively be encoded by a delta potential $2\eta\delta(r)$,
	and the resulting model for 1D intermediate statistics, 
	represented with a modified kinetic energy operator 
	$$
		\hat{T}_{\uS} 
		:= -\frac{1}{2}\sum_{j=1}^N \frac{\partial^2}{\partial x_j^2}
		+ \sum_{1\le j < k \le N} 2\eta \,\delta(x_k - x_j)
	$$
	and acting on completely symmetric wave functions 
	$\psi \in \bigotimes^N_{\textup{sym}} L^2(\R)$,
	is equivalent to the Lieb-Liniger model \cite{Lieb-Liniger:63}
	for a gas of bosons on $\R$ with delta interactions.
	
	Another, inequivalent but 
	just as well-motivated,
	approach to 1D intermediate statistics 
	is to consider \emph{Heisenberg-type quantization} 
	\cite{Leinaas-Myrheim:88,Polychronakos:89,Isakov:92,Leinaas-Myrheim:93} 
	in the sense that one starts out with the classical phase space 
	$\R^N \times \R^N$ and its observables
	(both should be symmetrized w.r.t.
	relabeling of particles)
	and considers representations of a corresponding algebra of operators.
	One finds, in the two-particle case, that the relevant representations
	are labeled by a continuous parameter 
	$\alpha \in \R$, and that 
	in the standard coordinate representation this amounts
	to imposing a boundary condition 
	$\psi(r) \sim r^\alpha$ for the wave function
	as the relative coordinate $r \to 0^+$
	(this relation with boundary conditions is more complicated
	and will be elaborated in Section \ref{sec:identical_1D} below, but
	the restriction $\alpha \in (-1/2,\infty)$ is always
	required for square-integrability).
	In this context bosons correspond to $\alpha=0$ 
	and fermions to $\alpha=1$.
	Again, using symmetric $\psi$ the general boundary condition 
	or choice of representation can 
	equivalently be encoded by an 
	interaction potential, of the form $\alpha(\alpha-1)/r^2$.
	The resulting model can be extended to the $N$-particle case
	and is equivalent to the homogeneous part 
	of the Calogero-Sutherland model 
	\cite{Calogero:69,Sutherland:71},
	represented with a modified kinetic energy operator
	$$
		\hat{T}_{\uH} 
		:= -\frac{1}{2}\sum_{j=1}^N \frac{\partial^2}{\partial x_j^2}
		+ \sum_{1\le j < k \le N} \frac{\alpha(\alpha-1)}{(x_k - x_j)^2},
	$$
	and again acting on completely symmetric wave functions 
	$\psi \in \bigotimes^N_{\textup{sym}} L^2(\R)$.

	One may ask what properties these different types of 
	identical particles possess.
	The fundamental characteristic property of fermions is 
	that they obey 
	\emph{Pauli's exclusion principle}, 
	i.e. the fact that the antisymmetry of the wave function
	implies that no two particles can occupy the same single-particle
	state: $\varphi \wedge \varphi = 0$ for any $\varphi \in L^2(\R^d)$.
	A powerful consequence of this is
	the \emph{Lieb-Thirring inequality} 
	\cite{Lieb-Thirring:75,Lieb-Thirring:76} (see also \cite{Lieb-Seiringer:10})
	for the energy of $N$ (spinless) 
	fermions in an external potential $V$ in $\R^d$, given by
	(again, with the conventions $\hbar = m = 1$, and $\|\psi\| = 1$)
	\begin{multline} \label{Lieb-Thirring}
		\int_{\R^{dN}} \sum_{j=1}^N \left( 
			\frac{1}{2} |\nabla_j \psi|^2 + V(\bx_j)|\psi|^2 
			\right) d\sx \\
		\ \ge \ -\sum_{k=0}^{N-1} |\lambda_k(h)|
		\ \ge \ -C_d \int_{\R^d} |V_-(\bx)|^{1 + \frac{d}{2}} \,d\bx,
	\end{multline}
	where $V_{\pm} := (V \pm |V|)/2$
	and $C_d$ a positive constant.
	Here $\lambda_k(h)$ denote the negative eigenvalues 
	(ordered by decreasing magnitude and with multiplicity) 
	of the one-particle operator
	$h = -\frac{1}{2}\Delta + V(\bx)$.
	The first inequality expresses the Pauli exclusion principle
	while the second concerns the trace over the negative spectrum of $h$
	and in some sense incorporates the uncertainty principle.
	In fact, the Lieb-Thirring inequality 
	is equivalent to the \emph{kinetic energy inequality}
	\begin{equation} \label{kinetic-energy-inequality}
		T_0 := \frac{1}{2} \int_{\R^{dN}} \sum_{j=1}^N |\nabla_j \psi|^2 \,d\sx
		\ \ge \ C_d' \int_{\R^d} \rho(\bx)^{1 + \frac{2}{d}} \,d\bx,
	\end{equation}
	with $C_d' := d(2/C_d)^{2/d}/(d+2)^{1+2/d}$ another 
	positive constant\footnote{The currently best known value
		for $C_d'$ is a factor $(3/\pi^2)^{1/d}$ smaller than its
		semiclassical value with
		$C_d = \frac{2^{d/2}}{(2\pi)^d} \int_{\R^d} (1-|\xi|^2)_+ \,d\xi$
		\cite{Dolbeault-Laptev-Loss:08}.
		It is conjectured that for $d=3$ the inequality 
		\eqref{kinetic-energy-inequality} for fermions holds with exactly the
		semiclassical constant and the Thomas-Fermi expression on the r.h.s.
		\cite{Lieb-Thirring:75}.}, 
	and $\rho$ the one-particle density function
	associated to $\psi$:
	\begin{equation} \label{one-particle-density}
		\rho(\bx) := \sum_{j=1}^N \int_{\R^{d(N-1)}} 
		|\psi(\bx_1,\ldots,\bx_{j-1},\bx,\bx_{j+1},\ldots,\bx_N)|^2 
		\prod_{k \neq j} d\bx_k.
	\end{equation}
	The inequality \eqref{kinetic-energy-inequality} 
	can be interpreted as a strong
	form of the uncertainty principle valid for fermions.
	
	These inequalities need to be weakened in the case of weaker exclusion.
	E.g., in the situation that up to
	$q \in \N$ particles can occupy the same one-particle state
	(sometimes referred to as \emph{Gentile intermediate statistics} 
	\cite{Gentile:40,Gentile:42}),
	the r.h.s. of \eqref{Lieb-Thirring} resp. \eqref{kinetic-energy-inequality}
	are to be multiplied by $q$ resp. $q^{-2/d}$.
	For bosons there is no restriction on 
	the number of particles in the same state;
	we can e.g. consider 
	$\psi = \varphi_0 \otimes \varphi_0 \otimes \ldots \otimes \varphi_0$,
	where $\varphi_0$ is the ground state of $h$.
	Hence bosons can be accommodated by taking $q=N$, 
	and the inequalities become trivial as $N \to \infty$.
	However,
	for intermediate and fractional statistics (in the above sense)
	the picture is more complicated.
	There have 
	been partially successful attempts to relate also such 
	many-body quantum states 
	to one-particle states restricted by some exclusion principle
	\cite{Haldane:91,Isakov:94,Goldin-Majid:04},
	but so far no general picture has emerged.
	The difficulty has to do with the fact that these generalized
	exchange statistics are naturally modeled using \emph{interacting}
	Hamiltonians, hence leaving the much simpler realm of 
	single-particle operators and spaces.

	In \cite{Lundholm-Solovej:anyon} we found that anyons 
	in $\R^2$, with statistics parameter
	$\alpha = \mu/\nu$ an \emph{odd numerator fraction},
	satisfy a kinetic energy inequality of the form
	\begin{equation} \label{TA-kinetic}
		T_{\uA} := \frac{1}{2} \int_{\R^{2N}} \sum_{j=1}^N |D_j \psi|^2 \,d\sx
		\ \ge \ \frac{1}{\nu^2} C_{\uA} \int_{\R^2} \rho(\bx)^2 \,d\bx,
	\end{equation}
	implying the Lieb-Thirring inequality
	\begin{equation} \label{TA-LT}
		\sum_{j=1}^N \int_{\R^{2N}} \left( \frac{1}{2} |D_j \psi|^2 
			+ V(\bx_j)|\psi|^2 \right) \,d\sx
		\ \ge \ -\nu^2 \,C_{\uA}' \int_{\R^2} |V_-(\bx)|^2 \,d\bx,
	\end{equation}
	for some positive constants $C_{\uA}$, $C_{\uA}'$
	(numerical estimates will be given below).
	For a general statistics parameter $\alpha \in \R$ and finite $N$ these 
	inequalities hold with $1/\nu^2$ replaced by $C_{\alpha,N}^2$, 
	where
	\begin{equation} \label{C_alpha}
		C_{\alpha,N} := \min_{p \in \{0,1,\ldots,N-2\}} \min_{q \in \Z} 
			|(2p+1)\alpha - 2q|.
	\end{equation}
	However, this expression tends to zero as $N \to \infty$,
	except when $\alpha = \mu/\nu$ is an odd numerator (reduced) fraction
	in which case $\lim_{N \to \infty} C_{\alpha,N} = 1/\nu$
	(cp. Figure 2 in \cite{Lundholm-Solovej:anyon}).
	This, at first unexpected, difference 
	between odd numerator fractions 
	and even numerator and irrational $\alpha$
	is further discussed in \cite{Lundholm-Solovej:exclusion},
	where we also provide some arguments 
	for this difference actually being natural.
	
	In contrast to previous proofs of the standard Lieb-Thirring inequalities
	\eqref{Lieb-Thirring} and \eqref{kinetic-energy-inequality},
	the inequality \eqref{TA-kinetic} for anyons was proved using
	a \emph{local} form of the uncertainty principle 
	together with a \emph{local} form of the exclusion principle, 
	valid for anyons with general $\alpha$ and involving
	the constant $C_{\alpha,N}$ in \eqref{C_alpha}
	measuring the strength of exclusion.
	Our inspiration for this local approach to Lieb-Thirring inequalities,
	applicable to interacting systems, 
	was 
	the original proof of stability of ordinary fermionic matter 
	\cite{Dyson-Lenard:67} due to Dyson and Lenard
	(see also \cite{Dyson:68,Lenard:73})
	which involved only local, comparatively weak, 
	consequences of the Pauli principle
	for fermions (see Lemma \ref{lem:local_exclusion_F} below).

\subsection{Main results}
	
	In this work we consider 
	further consequences of our results in 
	\cite{Lundholm-Solovej:anyon} for anyons, and
	extend the above families of Lieb-Thirring-type 
	inequalities to identical particles obeying intermediate
	exchange statistics in one dimension, in the above sense.
	We again start 
	from a local form of the uncertainty principle 
	together with a local exclusion principle of the form
	(cp. 
	Lemma \ref{lem:local_exclusion_F} below
	for fermions)
	\begin{equation} \label{local_exclusion_sketch}
		\int\limits_{[a,b]^n} \overline{\psi(\sx)} \left( \hat{T}_{\uS/\uH} \psi \right)(\sx) \,d\sx
			\ \ge \ (n-1)\frac{\xi_{\uS/\uH}^2}{|a-b|^2} \int\limits_{[a,b]^n} |\psi(\sx)|^2 \,d\sx,
	\end{equation}
	where in the Lieb-Liniger case 
	$\xi_{\uS}$ is a measure of exclusion 
	depending on $\eta$ times the length $|a-b|$ of a local interval,
	and in the Calogero-Sutherland case $\xi_{\uH}$ depends only on $\alpha$
	(see Figure \ref{fig:xiS} resp. \ref{fig:xiH} below
	for their exact dependence).

	\begin{thm} \label{thm:main}
		Let $\rho$ be the one-particle density 
		\eqref{one-particle-density}
		associated to 
		a normalized completely symmetric wave function 
		$\psi \in \bigotimes^N_{\textup{sym}} L^2(\R)$
		of $N \ge 2$ identical particles on the real line.

		For the Lieb-Liniger case, with statistics parameter $\eta \ge 0$
		and total kinetic energy 
		$T_{\uS} = \langle \psi, \hat{T}_{\uS}\psi \rangle$,
		we have
		\begin{equation} \label{TS-kinetic}
			T_{\uS} \ \ge \ 
			C_{\uS} \int_{\R} \xi_{\uS}(2\eta/\rho^*(x))^2 \rho(x)^3 \,dx,
		\end{equation}
		where 
		$\rho^*$ is
		the associated Hardy-Littlewood maximal function 
		(see \eqref{Hardy-Littlewood-rho}),
		and $C_{\uS}$ is a
		universal constant satisfying
		$3 \cdot 10^{-5} \le C_{\uS} \le 2/3$.

		For the Calogero-Sutherland case, with statistics parameter $\alpha \ge 1$
		and total kinetic energy 
		$T_{\uH} = \langle \psi, \hat{T}_{\uH}\psi \rangle$,
		and with its restriction to an arbitrary finite interval $Q \subseteq \R$ 
		denoted $T_{\uH}^Q$ (see \eqref{local_kinetic_energy_H}),
		we have whenever $\int_Q \rho \ge 2$ that
		\begin{equation} \label{TH-kinetic}
			T_{\uH}^Q \ \ge \ 
			C_{\uH} \, \xi_{\uH}(\alpha)^2 \frac{(\int_Q \rho \,dx)^3}{|Q|^2},
		\end{equation}
		for a 
		universal constant $1/32 \le C_{\uH} \le 2/3$.
	\end{thm}
	
	In particular, 
	if the density $\rho$ is confined to an interval of length $L$, 
	then 
	\eqref{TH-kinetic} 
	implies for the energy per unit length
	\begin{equation} \label{TH-thermodynamics}
		T_{\uH}/L \ \ge \ C_{\uH} \,\xi_{\uH}(\alpha)^2 \bar{\rho}^3,
	\end{equation}
	where $\bar{\rho} := N/L$ is the average density of particles.
	Note that $\xi_{\uH}(\alpha) \sim \alpha$ to leading order as
	$\alpha \to \infty$ (see Figure \ref{fig:xiH}).
	We can compare this bound with 
	Calogero \cite{Calogero:69} and Sutherland \cite{Sutherland:71}
	(whose models are exactly solvable
	for certain choices of external potentials),
	where one finds 
	for the ground state energy of a system confined to an interval 
	\begin{equation} \label{CS-thermodynamics}
		T_{\uH}/L \to \frac{\pi^2}{6} \alpha^2 \bar{\rho}^3,
	\end{equation}
	in the thermodynamic limit $N,L \to \infty$
	with fixed density $\bar{\rho}$.
	However, note that we do not need to assume a particular 
	confinement potential 
	for our bounds since they concern
	the kinetic energy $T_{\uH}$ alone.
	The possibility for considering more general external potentials
	in the Calogero-Sutherland models has also been discussed in the
	context of Thomas-Fermi theory; see
	\cite{Sen-Bhaduri:95,Smerzi:96}.
	
	In the Lieb-Liniger case,
	the bound \eqref{TS-kinetic} implies that if 
	$\rho^* \le \gamma \bar{\rho}$ for some constant
	$\gamma>0$, i.e. if the density is sufficiently homogeneous, then
	\begin{equation} \label{TS-thermodynamics}
		T_{\uS}/L \ \ge \ C_{\uS} \,\xi_{\uS}(2\eta/(\gamma\bar{\rho}))^2 \bar{\rho}^3.
	\end{equation}
	We have $\xi_{\uS}(t) \sim \sqrt{t}$ for small $t$ and
	$\xi_{\uS}(t) \to \pi/2$ as $t \to \infty$ (see Figure \ref{fig:xiS}).
	Compare with Lieb and Liniger \cite{Lieb-Liniger:63}
	(whose model is again exactly solvable in the absence 
	of an external potential),
	where it is shown that in the thermodynamic limit 
	the ground state energy satisfies
	\begin{equation} \label{LL-thermodynamics}
		T_{\uS}/L \to \frac{1}{2} e(2\eta/\bar{\rho}) \bar{\rho}^3,
	\end{equation}
	with an implicitly defined function $e$ s.t.
	$e(t) \sim t, t \ll 1$, $e(t) \to \frac{\pi^2}{3}, t \to \infty$.
	We can also compare with \cite{Lieb-Seiringer-Yngvason:03},
	where an energy functional of a form similar to the r.h.s. of 
	\eqref{TS-kinetic} arises in the limit of tubular confinement
	of a three-dimensional bosonic gas.
	
	The methods we use to prove Theorem \ref{thm:main}
	are similar to those used for anyons in \cite{Lundholm-Solovej:anyon}. 
	However,
	certain new technical complications arise 
	in the one-dimensional context,
	such as a local dependence of the strength of exclusion
	in the Lieb-Liniger case, and the possibility of arbitrarily 
	strong exclusion in the Calogero-Sutherland case.
	This is the reason for the more complicated expressions 
	\eqref{TS-kinetic} and \eqref{TH-kinetic}
	as compared to \eqref{kinetic-energy-inequality} and \eqref{TA-kinetic}.
	There is certainly room for improvement in our bounds for the constants
	$C_{\uA/\uS/\uH}$ (we mainly consider the forms of the inequalities 
	and their dependence on the statistics parameters
	to be of conceptual interest), 
	and for our methods we need to restrict to nonnegative
	statistics potentials, i.e. $\eta \ge 0$ resp. $\alpha \ge 1$.
	The intermediate Calogero-Sutherland case $\alpha \in (0,1)$
	is certainly interesting (so is $\alpha \in (-1/2,0)$), but
	presents additional challenges and will not be addressed here.
	
	In Section \ref{sec:identical_particles} we establish 
	the context and notation of the paper,
	taking care to define the relevant operators properly.
	This is something we have not found been discussed in detail
	in the literature.
	In particular, we show that the formal definition of the kinetic energy
	form used for anyons in \cite{Lundholm-Solovej:anyon}
	is indeed the natural one. 
	In Section \ref{sec:local_exclusion} we recall the 
	local exclusion principle for fermions and anyons,
	with an application of the latter to an explicit bound for the 
	energy of the ideal anyon gas.
	We then introduce corresponding local principles of exclusion
	for 1D intermediate statistics.
	In Section \ref{sec:local_uncertainty} we deduce a local form 
	of the uncertainty principle valid in arbitrary dimensions.
	These local bounds are then applied in Section \ref{sec:Lieb-Thirring}
	to prove Theorem \ref{thm:main}.
	In the final section we
	consider applications of the Lieb-Thirring 
	inequalities \eqref{TA-kinetic} and \eqref{TA-LT} 
	for anyons to the old 
	problem of
	many anyons confined in a harmonic oscillator potential,
	as well as the question of thermodynamic stability 
	for a system of charged anyons 
	and static particles interacting via 3D Coulomb potentials.
	We end with a discussion on the addition of external potentials 
	for the Calogero-Sutherland case.
	
	\begin{ack}
		We thank Giovanni Felder, J\"urg Fr\"ohlich, Jens Hoppe, 
		Edwin Langmann and Robert Seiringer 
		for comments and discussions.
		Support from the Danish Council for Independent Research 
		as well as from Institut Mittag-Leffler (Djursholm, Sweden)
		is gratefully acknowledged. D.L. would also like to thank
		IH\'ES, FIM ETH Zurich, and the Isaac Newton Institute 
		(EPSRC Grant EP/F005431/1)
		for support and hospitality via an EPDI fellowship,
		during which much of the present work was initiated.
	\end{ack}

\section{Identical particles in one and two dimensions} \label{sec:identical_particles}

	We refer to the brief introduction in \cite{Lundholm-Solovej:anyon},
	the original references 
	\cite{Streater-Wilde:70,Leinaas-Myrheim:77,Goldin-Menikoff-Sharp:81,Wilczek:82}, 
	and the reviews 
	\cite{Froehlich:90,Wilczek:90,Lerda:92,Myrheim:99,Polychronakos:99,Khare:05,Ouvry:07},
	for a complete introduction to the general concepts of identical particles
	and exchange statistics in one and two dimensions.
	Here we will jump directly to the consequences of the general theory 
	(for scalar particles) 
	outlined in the introduction,
	and make precise our mathematical assumptions and notation.

	As usual, we write $H^k(\Omega)$ for the Sobolev spaces of square-integrable
	functions on $\Omega \subseteq \R^n$
	with square-integrable weak derivatives to order $k$.
	For a function space $\mathcal{F}$
	we generally write $\mathcal{F}(X;Y)$ 
	to emphasize that the functions map $X \to Y$,
	suppressing the latter argument whenever $Y=\C$ 
	or if otherwise understood from the context.
	The space of smooth and compactly supported functions on $\Omega$
	is denoted by $C^\infty_c(\Omega)$ and its closure inside 
	$H^k(\Omega)$ by $H^k_0(\Omega)$.
	Given a space $\mathcal{F}$ of functions on the (traditional) 
	$N$-particle configuration space $X=(\R^d)^N$, 
	the subspace of functions which are completely symmetric 
	resp. antisymmetric w.r.t. particle permutations will be denoted
	$\mathcal{F}_{\textup{sym}}$ resp. $\mathcal{F}_{\textup{asym}}$.
	The domain of an operator $A$ is denoted by $\Dc(A)$.

\subsection{One dimension} \label{sec:identical_1D}

	Depending on which approach one takes to quantization 
	\cite{Leinaas-Myrheim:77,Leinaas-Myrheim:88,Polychronakos:89,
	Aneziris_et_al:91,Isakov:92,Leinaas-Myrheim:93,Isakov:94,Myrheim:99},
	identical particles in 1D can be modeled as bosons,
	i.e. wave functions symmetric under the flip 
	$r \mapsto -r$ of any two relative particle coordinates
	$r := x_j - x_k$,
	together with a local interaction potential, singular at $r=0$
	and either of the form $\delta_0(r)$ or $1/r^2$.
	We write 
	(cp. e.g. \cite{Myrheim:99})
	\begin{equation} \label{statistics_potentials}
		V_{\uS}(r) := 2\eta \delta_0(r),
		\qquad
		V_{\uH}(r) := \frac{\alpha(\alpha-1)}{r^2},
	\end{equation}
	with statistics parameters $\eta, \alpha \in \R$,
	for the corresponding cases of Sch\-r\"od\-ing\-er- 
	resp. Heisenberg-type quantization.
	These statistics potentials, 
	which coincide with the interaction potentials of the 
	Lieb-Liniger resp. Calogero-Sutherland models,
	should correspond to the choices of boundary conditions
	(in a sense to be made precise below)
	\begin{equation} \label{boundary_cond_S}
		\frac{\partial \psi}{\partial r} = \eta \psi, \quad \textrm{at $r=0^+$},
	\end{equation}
	resp.
	\begin{equation} \label{boundary_cond_H}
		\psi(r) \sim r^\alpha, \quad r \to 0^+,
	\end{equation}
	for the wave function $\psi$ at the boundary $r=0$ 
	of the configuration space.
	Note that in this sense $\eta=0$ resp. $\alpha=0$ represent bosons,
	while $\eta=+\infty$ resp. $\alpha=1$ represent fermions
	(in the bosonic representation, 
	i.e. after factoring out the sign of the permutation).
	Suggested by such pairwise boundary conditions, 
	one may \emph{define} (cp. e.g. \cite{Myrheim:99}) 
	the total kinetic energy for a normalized wave function 
	$\psi \in L^2_{\textup{sym}} := \bigotimes^N_{\textup{sym}} L^2(\R)$ 
	of $N$ identical particles on $\R$ to be
	$T_{\uS/\uH} := \int_{\R^N} T_{\uS/\uH}^N(\psi;\sx) \,d\sx$ where
	\begin{equation} \label{kinetic_energy_SH}
		T_{\uS/\uH}^N(\psi;\sx) 
		:= \frac{1}{2} \sum_{j=1}^N |\partial_j \psi|^2 
			+ \sum_{1\le j<k \le N} V_{\uS/\uH}(x_j - x_k) |\psi|^2,
	\end{equation}
	with a corresponding 
	kinetic energy operator
	\begin{equation} \label{kinetic_operator_SH}
		\hat{T}_{\uS/\uH} 
		= \hat{T}_0 + \hat{V}_{\uS/\uH}
		:= -\frac{1}{2} \sum_{j=1}^N \frac{\partial^2}{\partial x_j^2}
		+ \sum_{1\le j<k \le N} V_{\uS/\uH}(x_j - x_k).
	\end{equation}

	Formally, we need to specify domains 
	$\Dc_{\uS/\uH}$ resp. $\hat{\Dc}_{\uS/\uH}$
	for the quadratic forms 
	$\psi \mapsto T_{\uS/\uH}(\psi)$ 
	and operators $\hat{T}_{\uS/\uH}$
	so that they are closed, respectively self-adjoint.
	Let (for arbitrary dimension $d$)
	\begin{equation} \label{diagonals}
		\bDelta 
		:= \{ \sx \in (\R^d)^N : \text{$\exists\ j \neq k$ s.t. $\bx_j = \bx_k$} \}
	\end{equation}
	denote the diagonal set where any two particles meet,
	i.e. where the statistics potentials $V_{\uS/\uH}$ are singular.
	In the Lieb-Liniger case, 
	\begin{equation} \label{Schroedinger_form}
		T_{\uS}(\psi) = \frac{1}{2} \int_{\R^N} |\nabla \psi|^2 \,d\sx
		+ 2\eta \int_{\bDelta} |\psi|^2 \,d\Sigma_{\bDelta}
	\end{equation}
	(with $\Sigma_{\bDelta}$ the Euclidean measure on $\bDelta$)
	defines a closed and semibounded quadratic form on the domain
	$\Dc_{\uS} := H^1(\R^N) \cap L^2_{\textup{sym}}$,
	and hence also defines an associated self-adjoint operator 
	$\hat{T}_{\uS}$ on $L^2_{\textup{sym}}$.
	Its domain $\hat{\Dc}_{\uS}$ 
	is s.t. $\psi \in \hat{\Dc}_{\uS}$
	is twice (weakly) differentiable on $\R^N \setminus \bDelta$,
	continuous at the diagonals $\bDelta$,
	and satisfies (cp. also \cite{Lieb-Liniger:63})
	\begin{equation} \label{boundary_cond_S_ext}
		\left( \frac{\partial}{\partial x_j} - \frac{\partial}{\partial x_k} \right)\psi|_{x_j = x_k^+}
		- \left( \frac{\partial}{\partial x_j} - \frac{\partial}{\partial x_k} \right)\psi|_{x_j = x_k^-}
		= 4\eta\psi|_{x_j = x_k}.
	\end{equation}

	In the Calogero-Sutherland case, there is a complication in the choice of 
	domains associated with the symmetry $\alpha \mapsto 1-\alpha$.
	We therefore begin by considering the quadratic form 
	\begin{equation} \label{Heisenberg_form}
		T_{\uH}(\psi) = \frac{1}{2} \int_{\R^N} |\nabla \psi|^2 \,d\sx
		+ \alpha(\alpha-1) \int_{\R^N} \sum_{j<k} \frac{|\psi|^2}{|x_j-x_k|^2} \,d\sx,
	\end{equation}
	initially on the nice space of functions
	$\psi \in C_c^\infty(\R^N \setminus \bDelta) \cap L^2_{\textup{sym}}$.
	We can then use the substitution $\psi = f^\alpha \phi$,
	$f(\sx) := \prod_{j<k} |x_j-x_k|$,
	and a partial integration to write (see e.g. \cite{Lundholm})
	\begin{equation} \label{Hardy-GSR}
		\int_{\R^N} |\nabla \psi|^2 \,d\sx = \int_{\R^N} \left(
			\alpha(1-\alpha)\frac{|\nabla f|^2}{f^2} + \alpha\frac{-\Delta f}{f}
		\right) |\psi|^2 \,d\sx 
		+ \int_{\R^N} |\nabla \phi|^2 f^{2\alpha} \,d\sx,
	\end{equation}
	where $\Delta f=0$ on $\R^N \setminus \bDelta$
	and $|\nabla f|/f^2 = 2\sum_{j<k} |x_j-x_k|^{-2}$.
	Hence we have the representation
	$T_{\uH}(\psi) = \frac{1}{2} \|Q_\alpha \psi\|^2 \ge 0$ 
	for such $\psi$ and all $\alpha \in \R$,
	where $Q_\alpha$ denotes the vector-valued differential expression
	(cp. also \cite{Polychronakos:92})
	$$
		Q_\alpha := f^\alpha \nabla f^{-\alpha} = \nabla - \alpha W,
		\qquad W_j(\sx) := \sum_{k \neq j} (x_j - x_k)^{-1}.
	$$
	Now, consider 
	$Q_\alpha: L^2(\R^N) \to \mathcal{D}'(\R^N \setminus \bDelta; \C^N)$ 
	as a distribution-valued operator 
	(note that since $W \notin L^1_{\textup{loc}}(\R^N)$ 
	we need to be careful and remove the diagonals),
	and define $T_{\uH}$ with domain 
	$\Dc_{\uH} := \Dc(Q_\alpha^\textup{max})$ 
	to be the quadratic form 
	$T_{\uH}(\psi) := \frac{1}{2} \|Q_\alpha^\textup{max} \psi\|^2$
	associated to the \emph{maximal} extension of
	the operator $Q_\alpha$ on $L^2_{\textup{sym}}$, i.e.
	$$
		\Dc(Q_\alpha^\textup{max}) := \{ \psi \in L^2_{\textup{sym}} : Q_\alpha \psi \in L^2(\R^N; \C^N) \},
	$$
	with the identification 
	$L^2(\R^N; \C^N) = L^2(\R^N \setminus \bDelta; \C^N)$.
	Note that this definition of $T_{\uH}$ does not exclude
	the boundary behavior \eqref{boundary_cond_H} for $\alpha > -1/2$,
	taking e.g. $\psi(\sx) = f^\alpha(\sx) e^{-|\sx|^2} \in \Dc_{\uH}$,
	and that for $\alpha=0$, 
	$\Dc_{\uH} = H^1_{\textup{sym}}(\R^N \setminus \bDelta) = H^1_{\textup{sym}}(\R^N)$ 
	(note that the symmetry requirement is important here).
	The self-adjoint operator 
	$\hat{T}_{\uH} := \frac{1}{2} (Q_\alpha^{\textup{max}})^* Q_\alpha^{\textup{max}}$
	is then defined with domain 
	$$
		\hat{\Dc}_{\uH} := \{ \psi \in \Dc_{\uH} : Q_\alpha \psi \in \Dc((Q_\alpha^\textup{max})^*) \}.
	$$
	Another option would be to define $2\hat{T}_{\uH}$
	to be the Friedrichs extension, i.e. 
	$(Q_\alpha^{\textup{min}})^* Q_\alpha^{\textup{min}}$,
	associated to the minimally extended operator
	$Q_\alpha^{\textup{min}}$, its domain 
	$\Dc(Q_\alpha^\textup{min})$ being the closure
	of $C_c^\infty(\R^N \setminus \bDelta)$ inside $L^2_{\textup{sym}}$
	w.r.t. the form $\psi \mapsto \|Q_\alpha \psi\|^2$.
	In this case we can explicitly characterize its domain as
	$\Dc(Q_\alpha^\textup{min}) = H^1_{0,\textup{sym}}(\R^N \setminus \bDelta)$
	for $\alpha \neq 1/2$.
	This follows because for $\psi \in C_c^\infty(\R^N \setminus \bDelta)$ 
	the identity \eqref{Hardy-GSR} for $\alpha \neq 1/2$ and for $\alpha=1/2$
	(the Hardy inequality) implies
	$$
		c_1 \|\nabla \psi\| \le \|Q_\alpha \psi\| \le c_2 \|\nabla \psi\|
	$$
	for some constants $c_1,c_2>0$ (depending on $\alpha \neq 1/2$), 
	and hence taking the closure
	w.r.t. the $H^1$-form is the same as w.r.t. the $Q_\alpha$-form.

	On the other hand, the following shows that the two extensions 
	$Q_\alpha^{\textup{min}}$ and $Q_\alpha^{\textup{max}}$
	are the same for sufficiently large $\alpha$ 
	(while they are in general not for small $\alpha$,
	such as $\alpha=0$; 
	see also 
	\cite{Bruneau-Derezinski-Georgescu:11,Basu_et_al:03,Feher-Tsutsui-Fueloep:05}):

	\begin{thm} \label{thm:Heisenberg_domains}
		For $\alpha \ge 1$ we have 
		$Q_\alpha^{\textup{min}} = Q_\alpha^{\textup{max}}$,
		with domain $$\Dc_{\uH} = H^1_{0,\textup{sym}}(\R^N \setminus \bDelta),$$
		and hence the operator
		$\hat{T}_{\uH} = \frac{1}{2} (Q_\alpha^{\textup{max}})^* Q_\alpha^{\textup{max}}$
		is equal to the Friedrichs extension
		$\hat{T}_{\uH} = \frac{1}{2} (Q_\alpha^{\textup{min}})^* Q_\alpha^{\textup{min}}$.
		In the case $N=2$ we have 
		for all $\alpha > 1/2$ that
		$Q_\alpha^{\textup{min}} = Q_\alpha^{\textup{max}}$
		with $\Dc_{\uH} = H^1_{0,\textup{sym}}(\R^2 \setminus \bDelta)$.
	\end{thm}
	\begin{proof}
		We claim that, for $\alpha \ge 1$, 
		$\psi \in \Dc_{\uH}$ must satisfy $\psi \in H^1(\R^N)$ as well as
		\begin{equation} \label{W2-finite}
			\int_{\R^N} \frac{|\psi|^2}{(x_j-x_k)^2} \,d\sx < \infty \quad
			\text{for all $j\neq k$,}
		\end{equation}
		from which approximability in 
		$H^1_0(\R^N \setminus \bDelta) \cap L^2_{\textup{sym}}$ 
		follows by taking
		$\psi_\ep(\sx) := \prod_{j<k} \phi_\ep(x_j-x_k) \psi(\sx)$, where
		$\phi_\ep(x) := \phi(x/\ep)$ and $\phi \in C^\infty(\R;[0,1])$
		symmetric and identically zero for $|x| \le 1$ and one for $|x| \ge 2$.
		Then $\psi_\ep \in \Dc_{\uH}$ and
		$$
			Q_\alpha\psi_\ep - Q_\alpha\psi 
			= \left( {\textstyle\prod}\phi_\ep - 1 \right)Q_\alpha\psi 
			+ \left( \nabla {\textstyle\prod}\phi_\ep \right) \psi,
		$$
		where $\|(\prod\phi_\ep - 1)\psi\| \to 0$ and
		$\|(\prod\phi_\ep - 1)Q_\alpha\psi\| \to 0$ as $\ep \to 0$
		by dominated convergence, and for the second term we have by
		\eqref{W2-finite}
		\begin{multline*}
			\| \phi_\ep'(x_j-x_k) \psi \|^2 
			= \int_{\R^N} \frac{|\phi'((x_j-x_k)/\ep)|^2}{\ep^2} |\psi|^2 \,d\sx \\
			\le C \int_{\ep<|x_j-x_k|<2\ep} \frac{|\psi|^2}{(x_j-x_k)^2} \,d\sx \to 0.
		\end{multline*}
		Hence, $\|Q_\alpha (\psi_\ep - \psi)\| \to 0$ as $\ep \to 0$.
		
		To prove the above claim, consider the pointwise 
		a.e. in $\R^N \setminus \bDelta$ identity
		\begin{multline} \label{pointwise-Q_alpha}
			| Q_\alpha\psi |^2 
			= |\nabla\psi|^2 + \alpha^2 W \cdot W|\psi|^2 - \alpha W \cdot ((\nabla\bar{\psi}) \psi + \bar{\psi} \nabla\psi) \\
			= |\nabla\psi|^2 + \alpha^2 W^2 |\psi|^2 - \alpha W \cdot \nabla|\psi|^2  \\
			= |\nabla\psi|^2 + \alpha(\alpha-1) W^2 |\psi|^2 - \alpha \nabla \cdot (W|\psi|^2),
		\end{multline}
		where in the second step we used that
		$\psi \in H^1_{\textup{loc}}(\R^N \setminus \bDelta)$ since
		$W\psi \in L^2_{\textup{loc}}(\R^N \setminus \bDelta)$ and
		$Q_\alpha \psi = \nabla\psi - \alpha W\psi \in L^2(\R^N)$ 
		(note that then 
		$\partial_j|\psi|^2 - (\partial_j \bar\psi)\psi - \bar\psi(\partial_j\psi) = 0$ 
		in $\mathcal{D}'(\R^N \setminus \bDelta)$
		and $\nabla|\psi|^2 \in L^1_{\textup{loc}}(\R^N \setminus \bDelta)$),
		and in the third step we have
		$-\nabla \cdot W = W^2 = |\nabla f|/f^2 = 2\sum_{j<k}|x_j-x_k|^{-2}$.
		Now, for any $\ep \ge 0$, let us define the wedge-shaped set
		\begin{multline*}
			\Gamma_\ep := \{ \sx \in \R^N : x_1 + (N-1)\ep < x_2 + (N-2)\ep < \ldots < x_{N-1} + \ep < x_N \} \\
			= \bigcap_{j=1}^{N-1} \{ \sx \in \R^N : x_j + \ep < x_{j+1} \},
		\end{multline*}
		with piecewise flat boundary 
		$\partial \Gamma_\ep = \bigcup_{j=1}^{N-1} (\partial\Gamma_\ep)_{j,j+1}$,
		$$
			(\partial\Gamma_\ep)_{j,j+1} := \{ \sx \in \R^N : x_j + \ep = x_{j+1} \ \text{and}\ x_k + \ep \le x_{k+1} \ \text{for}\ k \neq j \}.
		$$
		By \eqref{pointwise-Q_alpha} 
		and the divergence theorem (see e.g. Theorem 6.9 in \cite{Lieb-Loss:01})
		we have for $\ep>0$
		\begin{equation} \label{Q_alpha_divergence}
			\int_{\Gamma_\ep} |Q_\alpha\psi|^2 \,d\sx
			= \int_{\Gamma_\ep} \left( 
				|\nabla\psi|^2 + \alpha(\alpha-1)W^2|\psi|^2 \right) d\sx
			- \alpha \int_{\partial \Gamma_\ep} \hat{n} \cdot W |\psi|^2 \,d\sx.
		\end{equation}
		On the boundary component $(\partial\Gamma_\ep)_{j,j+1}$,
		the unit outward normal is $\hat{n} = (e_j - e_{j+1})/\sqrt{2}$, 
		and we obtain pointwise
		\begin{multline*}
			\sqrt{2}\,\hat{n} \cdot W = W_j - W_{j+1} \\
			= 2(x_j - x_{j+1})^{-1} + \sum_{\substack{k=1 \\ k \neq j,j+1}}^N \left(
				(x_j - x_k)^{-1} - (x_{j+1} - x_k)^{-1}
			\right) \\
			= -2\ep^{-1} + \sum_{k \neq j,j+1} \ep(x_j-x_k)^{-1}(x_j-x_k+\ep)^{-1} \\
			= -2\ep^{-1} + \sum_{k < j} \ep(x_j-x_k)^{-1}(x_j-x_k+\ep)^{-1}
			 + \sum_{k > j+1} \ep(x_k-x_j)^{-1}(x_k-x_j-\ep)^{-1}.
		\end{multline*}
		For $N=2$ we then have 
		$-\alpha\,\hat{n} \cdot W = \frac{\alpha \ep^{-1}}{\sqrt{2}} 2$,
		while for $N>2$ the two additional positive sums above are bounded by
		$$
			\sum_{k=1}^{j-1} \ep ((j-k)\ep)^{-1} ((j-k)\ep + \ep)^{-1} 
			= \ep^{-1} \sum_{k=1}^{j-1} \frac{1}{k(k+1)} = \ep^{-1} \frac{j-1}{j},
		$$
		respectively
		$$
			\sum_{k=j+2}^{N} \ep ((k-j)\ep)^{-1} ((k-j)\ep - \ep)^{-1} 
			= \ep^{-1} \sum_{k=1}^{N-j-1} \frac{1}{k(k+1)} = \ep^{-1} \frac{N-j-1}{N-j},
		$$
		and hence
		$$
			-\alpha\,\hat{n} \cdot W 
			\ge \frac{\alpha \ep^{-1}}{\sqrt{2}} \left(
				2 - \frac{j-1}{j} - \frac{N-j-1}{N-j}
			\right)
			= \frac{\alpha \ep^{-1}}{\sqrt{2}} \frac{N}{j(N-j)}
			\ge \frac{\alpha \ep^{-1}}{\sqrt{2}} \frac{4}{N},
		$$
		which is strictly positive for all $N>2$.
		Hence, since all terms in the r.h.s. of \eqref{Q_alpha_divergence}
		are nonnegative for $\alpha \ge 1$, 
		and the l.h.s. remains finite as $\ep \to 0$,
		we must have $\nabla\psi, W\psi \in L^2(\Gamma_0)$
		(also in the case $\alpha=1$ we find $\nabla\psi \in L^2(\Gamma_0)$ 
		and therefore $W\psi \in L^2(\Gamma_0)$ since $Q_\alpha\psi \in L^2(\Gamma_0)$).
		By symmetry of $\psi$, this also holds on the remaining parts of 
		$\R^N \setminus \bDelta$ obtained by permutation of the coordinates,
		and this would then prove the claim if we only knew that 
		$\nabla\psi \in \mathcal{D}'(\R^N)$ 
		is really a function on all of $\R^N$ and not also having 
		some singular component supported on $\bDelta$.
		The support must furthermore be of codimension one 
		(cp. Lemma \ref{lem:folklore} below).
		But this would require a codimension one 
		discontinuity of $\psi$ at one of the boundary
		components $(\partial\Gamma_0)_{j,j+1}$, 
		which is in contradiction with both $W\psi \in L^2(\R^N)$ and
		the symmetry of $\psi$.
		This proves the first part of the theorem.

		It remains to prove 
		the stronger statement of the theorem for the 
		special case $N=2$ and $\alpha>1/2$ (this was also proved differently
		in \cite{Bruneau-Derezinski-Georgescu:11}, in a one-particle form).
		We start by showing that $|x_1-x_2|^{-\frac{1}{2}}\psi$ 
		is (transversally) uniformly bounded if $\alpha > 1/2$.
		Note that for $x_1<x_2$, 
		\begin{multline*}
			(x_2-x_1)^{-\alpha} \psi(\sx) 
			= -\int_0^\infty \frac{d}{dt}\left( (x_2-x_1 + 2t)^{-\alpha} \psi(x_1-t,x_2+t) \right) dt \\
			= -\int_0^\infty (x_2-x_1 + 2t)^{-\alpha} \Big( (\partial_2 - \partial_1)\psi (x_1-t,x_2+t) \\
				- 2\alpha (x_2-x_1+2t)^{-1}\psi(x_1-t,x_2+t) \Big) dt \\
			= -\int_0^\infty (x_2-x_1 + 2t)^{-\alpha} \left( (Q_{\alpha,2} - Q_{\alpha,1})\psi \right) (x_1-t,x_2+t) \,dt \\
			\le \left| \int_0^\infty (x_2-x_1 + 2t)^{-2\alpha} \,dt \right|^{\frac{1}{2}}
				\left| \int_0^\infty \left| (Q_{\alpha,2} - Q_{\alpha,1})\psi \right|^2 (x_1-t,x_2+t) \,dt \right|^{\frac{1}{2}} \\
			\le \left( \frac{(x_2-x_1)^{1-2\alpha}}{4\alpha-2} \right)^{\frac{1}{2}} 
				\|(Q_{\alpha,2} - Q_{\alpha,1})\psi\|_r,
		\end{multline*}
		where we used Cauchy-Schwarz, and $\|\cdot\|_r$ denotes the $L^2$-norm 
		w.r.t. the transversal variable $r=x_2-x_1$ 
		(but still depending on the longitudinal center-of-mass 
		$R=(x_1+x_2)/2$).
		Similarly for $x_1>x_2$, we have 
		$|r|^{-\frac{1}{2}}\psi(\sx) \le \|(Q_{\alpha,2} - Q_{\alpha,1})\psi\|_r/\sqrt{4\alpha-2}$.
		Taking $\psi_\ep(\sx) := \varphi_\ep(x_1-x_2)\psi(\sx)$,
		with $\varphi_\ep$ a logarithmic cut-off function defined 
		as in the proof of Lemma \ref{lem:folklore} below,
		we have
		$$
			Q_\alpha\psi_\ep - Q_\alpha\psi 
			= (\varphi_\ep - 1)Q_\alpha\psi \pm \varphi_\ep' \psi.
		$$
		The norm-squared of the last term is
		\begin{multline*}
			\ep^2 \int_R \int_r \frac{|\varphi'(\ep \ln |r|)|^2}{r^2} |\psi|^2 \,dr\,dR
			\le C \ep \int_R \|(Q_{\alpha,2} - Q_{\alpha,1})\psi\|_r^2 \,dR \\
			= C \ep \|(Q_{\alpha,2} - Q_{\alpha,1})\psi\|^2,
		\end{multline*}
		which shows that $\|Q_\alpha(\psi_\ep - \psi)\| \to 0$
		as $\ep \to 0$,
		with $\psi_\ep \in H^1_{0,\textup{sym}}(\R^2 \setminus \bDelta)$.
	\end{proof}

	We will in the following always assume 
	$\eta \ge 0$ resp. $\alpha \ge 1$
	in order for the statistics potentials 
	$V_{\uS}$ resp. $V_{\uH}$ to be nonnegative,
	and can hence work with the Friedrichs extension for $\hat{T}_{\uH}$.
	To emphasize the dependence on the statistics parameters,
	the corresponding domains will be denoted 
	$\hat{\Dc}_{\uS}^\eta \subseteq \Dc_{\uS}^\eta$
	resp.
	$\hat{\Dc}_{\uH}^\alpha \subseteq \Dc_{\uH}^\alpha$.
	It is well-known that in the exactly solvable Calogero-Sutherland
	models, involving the operator $\hat{T}_{\uH}$, the eigenfunctions
	satisfy the boundary condition \eqref{boundary_cond_H}
	(see \cite{Calogero:69,Sutherland:71}).
	However, for large $\alpha$ this requirement is too restrictive,
	since e.g. $\psi(\sx) = f^2(\sx)e^{-|\sx|^2}$ is both in 
	$\Dc_{\uH}^\alpha$ and $\hat{\Dc}_{\uH}^\alpha$,
	even for $\alpha > 2$.
	The explicit dependence on $\alpha$ of the domains 
	$\Dc_{\uH}^\alpha$ and $\hat{\Dc}_{\uH}^\alpha$
	for the operators appearing in the two-particle case $N=2$ has been
	thoroughly investigated in \cite{Bruneau-Derezinski-Georgescu:11}.
	Also the extension theory for the $N=3$ case has been considered in 
	detail in \cite{Feher-Tsutsui-Fueloep:05}.
	The complete behavior in the many-particle case for $\alpha \in (-1/2,1)$ 
	remains, to the best of our knowledge, an interesting open problem.

\subsection{Two dimensions} \label{sec:identical_2D}

	(Abelian) anyons in $\R^2$ comprise 
	a continuous family of identical particles, characterized by
	a one-dimensional unitary representation of the braid group $B_N$,
	i.e. a complex phase $e^{i\alpha\pi} \in U(1)$,
	or a real statistics parameter 
	$\alpha \in \R$ (modulo 2). 
	For $\alpha=0$ such particles are bosons while for $\alpha=1$ they are fermions.
	For general $\alpha$,
	one can model quantum mechanical wave functions $\psi$
	of $N$ such anyons by means of bosonic wave functions on $\R^2$, 
	i.e. completely symmetric functions in $L^2((\R^2)^N)$,
	together with magnetic interaction potentials
	$\bA_j := \alpha\sum_{k \neq j} (\bx_j - \bx_k)^{-1}I$
	of topological type between all particles.
	Here $\bx I$ denotes the $90^\circ$ counter-clockwise rotation of the vector 
	$\bx \in \R^2$, and $\bx^{-1} := \bx/|\bx|^2$.
	Hence, the total kinetic energy for $\psi$ is given by
	$T_{\uA} := \int_{\R^{2N}} T_{\uA}^N(\psi;\sx) \,d\sx$, 
	where
	\begin{equation} \label{kinetic_energy_A}
		T_{\uA}^N(\psi;\sx) := \frac{1}{2} \sum_{j=1}^N |D_j \psi|^2,
		\quad D_j := -i\nabla_j + \bA_j
			= -i\nabla_{\bx_j} + \alpha\sum_{\substack{k=1 \\ k \neq j}}^N (\bx_j - \bx_k)^{-1}I,
	\end{equation}
	corresponding to a kinetic energy operator
	\begin{equation} \label{kinetic_operator_A}
		\hat{T}_{\uA} 
		:= \frac{1}{2} \sum_{j=1}^N D_j \cdot D_j.
	\end{equation}
	There is also here a question on how this operator and 
	quadratic form should be formally defined 
	since $\bA_j(\sx)$ diverges badly as $\bx_k \to \bx_j$.
	In \cite{Lundholm-Solovej:anyon} we defined 
	the operator $\hat{T}_{\uA}$ in \eqref{kinetic_operator_A} 
	with a closed quadratic form $T_{\uA}$
	as the Friedrichs extension initially defined on 
	$C_c^\infty(\R^{2N} \setminus \bDelta) \cap L^2_{\textup{sym}}$.
	This might seem to imply a mild hard-core requirement for 
	the wave function for $\alpha \neq 0$ 
	(cp. e.g. \cite{Loss-Fu:91,Baker_et_al:93} where a stronger
	form of hard-core requirement for anyons was imposed).
	Let us first observe explicitly that this is not the case whenever
	$\alpha \in 2\Z$ (respecting the periodicity of the statistics parameter),
	and then show for general $\alpha \in \R$
	that this definition is truly non-restrictive 
	in the sense that it coincides with that of the \emph{maximal} 
	extension for the quadratic form $T_{\uA}$.
	In other words, if we insist on the kinetic energy form 
	\eqref{kinetic_energy_A} being finite, then we necessarily
	have this self-adjoint extension.
	
	We begin by defining the magnetic derivative as 
	the distribution-valued operator
	$$
		D: L^2(\R^{2N}) \to \mathcal{D}'(\R^{2N} \setminus \bDelta; \C^{2N}),
		\quad
		\psi \mapsto (D_j \psi)_{j=1}^N,
	$$
	where $D_j$ are the $\C^2$-vector-valued differential expressions in 
	\eqref{kinetic_energy_A}.
	Note that we again need to consider this as a 
	distribution strictly outside the diagonals since 
	$\bA_j \notin L^2_\textup{loc}(\R^{2N})$ 
	and hence $\bA_j\psi \notin L^1_\textup{loc}(\R^{2N})$,
	although $\bA_j \in C^\infty(\R^{2N} \setminus \bDelta)$ 
	and $\nabla \psi \in \mathcal{D}'(\R^{2N})$.
	We have a corresponding maximal domain for the operator 
	$D$ acting on the Hilbert space $L^2_{\textup{sym}}(\R^{2N})$,
	which defines the form $T_{\uA}$:
	$$
		\Dc_{\uA} := \Dc(D^\textup{max}) := \{ \psi \in L^2_{\textup{sym}} : D \psi \in L^2(\R^{2N}; \C^{2N}) \},
		\quad T_{\uA}(\psi) = \frac{1}{2} \|D\psi\|^2.
	$$
	This also defines an associated self-adjoint operator
	$\hat{T}_{\uA} := \frac{1}{2} (D^{\textup{max}})^* D^{\textup{max}}$ 
	with domain
	$$
		\hat{\Dc}_{\uA} := \{ \psi \in \Dc_{\uA} : D\psi \in \Dc((D^\textup{max})^*) \}.
	$$
	To emphasize the dependence on $\alpha$, 
	the corresponding domains will be denoted 
	$\hat{\Dc}_{\uA}^\alpha \subseteq \Dc_{\uA}^\alpha$.

	Again, another option is to consider the Friedrichs extension,
	i.e. define the minimal extension
	$D^{\textup{min}}$ as the closure of $D$ acting on 
	$C_c^\infty(\R^{2N} \setminus \bDelta)$,
	$$
		\Dc(D^\textup{min}) 
		:= \overline{C_c^\infty(\R^{2N} \setminus \bDelta) \cap L^2_{\textup{sym}}}
		\ \subseteq \Dc_{\uA},
	$$
	where the closure is taken w.r.t. the graph norm in 
	$L^2_{\textup{sym}} \times L^2$ 
	and the form $\psi \mapsto \|D\psi\|^2$.
	The associated self-adjoint operator is
	$\hat{T}_{\uA}^\textup{min} := \frac{1}{2} (D^{\textup{min}})^* D^{\textup{min}}$.
	Note first that for $\alpha=0$, 
	i.e. for the free kinetic energy for bosons, 
	these two definitions coincide, 
	$\hat{T}_{\uA} = \hat{T}_{\uA}^\textup{min} = \hat{T}_0$, 
	as a consequence of the following lemma. 

	\begin{lem} \label{lem:folklore}
		We have equality for the Sobolev spaces 
		$H^1(\R^{2N}) = H^1_0(\R^{2N} \setminus \bDelta)$,
		$H^1_{\textup{sym}}(\R^{2N}) = H^1_{0,\textup{sym}}(\R^{2N} \setminus \bDelta)$,
		$H^1_{\textup{asym}}(\R^{2N}) = H^1_{0,\textup{asym}}(\R^{2N} \setminus \bDelta)$,
		as well as $H^1(\R^{2N}) = H^1(\R^{2N} \setminus \bDelta)$.
	\end{lem}
	\begin{proof}
		We first show that 
		$H^1_{(\textup{sym/asym})}(\R^{2N}) = H^1_{0,(\textup{sym/asym})}(\R^{2N} \setminus \bDelta)$.
		Since $C_c^\infty(\R^{2N}) \cap L^2_{(\textup{sym/asym})}$ 
		is dense in $H^1(\R^{2N}) \cap L^2_{(\textup{sym/asym})}$ 
		we can assume that $\psi \in C_c^\infty(\R^{2N})$.
		Define $\psi_\ep(\sx) := \prod_{j<k} \varphi_\ep(\bx_j-\bx_k) \psi(\sx)$
		for $\ep>0$, where $\varphi_\ep(\bx) := \varphi(\ep \ln|\bx|)$,
		and $\varphi \in C^\infty(\R;[0,1])$ 
		is taken to be identically zero on $(-\infty,-2)$
		and one on $(-1,\infty)$.
		Then $\psi_\ep \in C_c^\infty(\R^{2N} \setminus \bDelta) \cap L^2_{(\textup{sym/asym})}$,
		$$
			\|\psi_\ep - \psi\|^2 
			= \int_{\R^{2N}} \left| {\textstyle\prod} \varphi_\ep - 1 \right|^2 |\psi|^2 \,d\sx
			\to 0, \ \ep \to 0,
		$$
		since the measure of $\supp (\prod\varphi_\ep-1) \cap \supp \psi$ 
		tends to zero, and
		\begin{multline*}
			\|\nabla(\psi_\ep - \psi)\| 
			= \| ({\scriptstyle\prod} \varphi_\ep - 1)\nabla\psi 
				+ (\nabla {\scriptstyle\prod} \varphi_\ep)\psi \| \\
			\le \| ({\scriptstyle\prod} \varphi_\ep - 1)\nabla\psi \|_{L^2(\supp\psi)}
			+ \| \nabla {\scriptstyle\prod} \varphi_\ep \|_{L^2(\supp\psi)} \|\psi\|_{L^\infty}.
		\end{multline*}
		The first term on the r.h.s. tends to zero as above, 
		while for the second we have for each $j \in \{1,\ldots,N\}$
		$$
			\left\| \nabla_j {\textstyle\prod_{k<l}} \varphi_\ep(\bx_k - \bx_l) \right\|_{L^2(\supp\psi)}
			\ \le \ \sum_{k\neq j} \| \nabla_j \varphi_\ep(\bx_j - \bx_k) \|_{L^2(\supp\psi)}
		$$
		and
		$$
			\| \nabla_j \varphi_\ep(\bx_j - \bx_k) \|_{L^2(\supp\psi)}^2
			\ \le \ \int_{\supp\psi} \int_{\R^2} |\nabla \varphi_\ep(\bx) |^2 \,d\bx \prod_{l\neq j} d\bx_l.
		$$
		Finally note that, by the substitution $r = e^s$,
		$$
			\frac{1}{2\pi} \int_{\R^2} |\nabla \varphi_\ep|^2 \,d\bx 
			= \ep^2 \int_0^\infty |\varphi'(\ep\ln r)|^2 \,\frac{dr}{r}
			= \ep^2 \int_{-2/\ep}^{-1/\ep} |\varphi'(\ep s)|^2 \,ds \le C \ep,
		$$
		which proves the first part of the lemma.
		
		It remains to show that we can identify
		$H^1(\R^{2N}) = H^1(\R^{2N} \setminus \bDelta)$.
		Any element of the l.h.s. can obviously also be interpreted uniquely
		as an element of the r.h.s.
		Conversely, 
		assume that $\psi \in L^2(\R^{2N})$ is s.t.
		the locally defined distribution 
		$\nabla \psi \in \mathcal{D}'(\R^{2N} \setminus \bDelta)$
		is actually a 
		function on $\R^{2N} \setminus \bDelta$ 
		which extends to be square-integrable on $\R^{2N}$.
		Denote this function by $\mu \in L^2(\R^{2N})$.
		To see that this function is identical to the 
		canonically defined distribution
		$\nabla\psi \in H^{-1}(\R^{2N})$ 
		when both are considered as distributions on the full space $\R^{2N}$,
		let $f := \nabla\psi - \mu \in H^{-1}(\R^{2N})$ and observe 
		using the first part of the lemma
		that for any $u \in H^1(\R^{2N}) = H^1_0(\R^{2N} \setminus \bDelta)$,
		the dual space of $H^{-1}(\R^{2N})$, 
		with $u \leftarrow u_n \in C_c^\infty(\R^{2N} \setminus \bDelta)$,
		$$
			\langle f, u \rangle = \lim_{n \to \infty} \langle f, u_n \rangle = 0,
		$$
		since the distribution $f$ is supported on $\bDelta$.
		Hence $f=0$ in $\mathcal{D}'(\R^{2N})$.
	\end{proof}
	
	Since for $\alpha=0$, by definition
	$\Dc(D^\textup{max}) = H^1_{\textup{sym}}(\R^{2N} \setminus \bDelta)$ 
	and $\Dc(D^\textup{min}) \\= H^1_{0,\textup{sym}}(\R^{2N} \setminus \bDelta)$,
	we have by the above lemma that $D^\textup{min} = D^\textup{max}$ 
	and hence obtain the natural kinetic energy operator 
	$\hat{T}_0 = -\frac{1}{2}\Delta$ for free bosons
	with domain $\hat{\Dc}_{\uA}^{\alpha=0} = H^2_{\textup{sym}}(\R^{2N})$.
	Now, define the unitary multiplication operator
	$U: L^2(\R^{2N}) \to L^2(\R^{2N})$,
	$$
		(U\psi)(\sx) := \prod_{j<k} \frac{z_j-z_k}{|z_j-z_k|} \psi(\sx) 
		= \prod_{j<k} e^{i\phi_{jk}} \psi(\sx),
		\qquad \sx \in \R^{2N} \setminus \bDelta,
	$$
	where $\phi_{jk} := \arg(z_j - z_k)$, $z_j := x_{j,1} + ix_{j,2}$.
	We then find that 
	$U: L^2_\textup{sym} \to L^2_\textup{asym}$,
	$L^2_\textup{asym} \to L^2_\textup{sym}$ 
	is an isomorphism and furthermore, 
	noting that $U$ is smooth on $\R^{2N} \setminus \bDelta$,
	we have
	$D^{(\alpha=n+\beta)} = U^{-n} D^{(\alpha=\beta)} U^{n}$
	as distribution-valued operators 
	for $n \in \Z$ and $\beta \in \R$.
	In particular,
	$$
		D^{(\alpha=2n)} = U^{-2n} (-i\nabla) U^{2n}
		\quad \ \text{and} \quad 
		\hat{T}_{\uA}^{(\alpha=2n)} = U^{-2n} \hat{T}_0 U^{2n},
	$$
	with $\|D^{(\alpha=2n)}\psi\| = \|\nabla(U^{2n}\psi)\|$,
	and therefore
	$$
		\Dc_{\uA}^{\alpha=2n} = U^{-2n} H^1_{\textup{sym}}(\R^{2N})
		\quad \text{and} \quad 
		\hat{\Dc}_{\uA}^{\alpha=2n} = U^{-2n} H^2_{\textup{sym}}(\R^{2N}).
	$$
	Similarly,
	$$
		\Dc_{\uA}^{\alpha=2n+1} = U^{-(2n+1)} H^1_{\textup{asym}}(\R^{2N})
		\quad \text{and} \quad 
		\hat{\Dc}_{\uA}^{\alpha=2n+1} = U^{-(2n+1)} H^2_{\textup{asym}}(\R^{2N}).
	$$
	This also shows that we might have $\nabla\psi \notin L^2(\R^{2N})$
	although $D\psi \in L^2(\R^{2N})$.

	The case of general $\alpha \in \R$ is a little trickier.
	Since $\bA_j \notin L^2_\textup{loc}(\R^{2N})$ we cannot directly
	apply the standard theorems for magnetic forms
	(see e.g. Theorems 7.21 and 7.22 in \cite{Lieb-Loss:01},
	\cite{Erdoes-Vougalter:02} for a related discussion,
	or \cite{DFT:97} for a non-magnetic approach), 
	but they can nevertheless be extended to
	the anyonic case as follows.
	
	\begin{lem}[Diamagnetic inequality] \label{lem:diamagnetic_inequality}
		Assume $\psi \in L^2(\R^{2N})$ and \\
		$D\psi \in L^2(\R^{2N}; \C^{2N})$
		for arbitrary fixed $\alpha \in \R$.
		Then $|\psi| \in H^1(\R^{2N})$, 
		\begin{equation} \label{diamag_pointwise}
			\big|\nabla_j |\psi|(\sx)\big| \le 
			|(D_j \psi)(\sx)|
		\end{equation} 
		for all $j \in \{1,\ldots,N\}$ and a.e. $\sx \in \R^{2N}$,
		and 
		\begin{equation} \label{diamag_norms}
			\| \nabla|\psi| \|_{L^2(\R^{2N}; \C^{2N})} \le \| D\psi \|_{L^2(\R^{2N}; \C^{2N})}.
		\end{equation} 
	\end{lem}
	\begin{proof}
		Note that $|\psi| \in L^2(\R^{2N})$ trivially.
		By the assumptions we have that $D\psi$, 
		considered as a distribution acting on 
		$C_c^\infty(\R^{2N} \setminus \bDelta)$,
		is actually a function locally on $\R^{2N} \setminus \bDelta$
		and which is square-integrable on all of $\R^{2N}$.
		In particular, 
		$D\psi = -i\nabla\psi + A\psi \in L^2_{\textup{loc}}(\R^{2N} \setminus \bDelta)$,
		where also $A\psi \in L^2_{\textup{loc}}(\R^{2N} \setminus \bDelta)$,
		and hence $-i\nabla\psi \in L^2_{\textup{loc}}(\R^{2N} \setminus \bDelta)$,
		i.e. $\psi \in H^1_{\textup{loc}}(\R^{2N} \setminus \bDelta)$.
		By Theorem 6.17 in \cite{Lieb-Loss:01} we then have
		$|\psi| \in H^1_{\textup{loc}}(\R^{2N} \setminus \bDelta)$
		and, for a.e. $\sx \in \R^{2N} \setminus \bDelta$,
		$$
			(\nabla|\psi|)(\sx) = \left\{ \begin{array}{ll}
				\re \left( \frac{\bar{\psi}}{|\psi|} \nabla\psi \right)(\sx), \ & \text{if $\psi(\sx) \neq 0$,} \\
				0, & \text{if $\psi(\sx) = 0$.}
			\end{array}\right.
		$$
		Now, for $j \in \{1,\ldots,N\}$ and each such point $\sx$ s.t. 
		$\psi(\sx)\neq 0$ we have
		$$
			\re\left( \frac{\bar{\psi}}{|\psi|} \nabla_j \psi \right)
			= \re\left( \frac{\bar{\psi}}{|\psi|} \nabla_j \psi + i\bA_j \frac{|\psi|^2}{|\psi|} \right)
			= \re\left( i\frac{\bar{\psi}}{|\psi|} (-i\nabla_j \psi + \bA_j \psi) \right),
		$$
		so
		$$
			\big| \nabla_j |\psi| \big|
			\le \left| i\frac{\bar{\psi}}{|\psi|} (-i\nabla_j \psi + \bA_j \psi) \right|
			= \big|{-i\nabla_j \psi + \bA_j \psi}\big|,
		$$
		while for $\psi(\sx)=0$,
		$(\nabla|\psi|)(\sx) = 0 \le |(-i\nabla_j \psi + \bA_j \psi)(\sx)|$.
		Hence, \eqref{diamag_pointwise} holds
		for a.e. $\sx \in \R^{2N} \setminus \bDelta$
		(i.e. a.e. $\sx \in \R^{2N}$), and
		$$
			\int_{\R^{2N}} \big| \nabla_j|\psi|(\sx) \big|^2 \,d\sx
			\le \int_{\R^{2N}} \big|(-i\nabla_j \psi + \bA_j \psi)(\sx)|^2 \,d\sx.
		$$
		We have shown that 
		$|\psi| \in H^1_{\textup{loc}}(\R^{2N} \setminus \bDelta)$
		with $\int_{\R^{2N}} |\nabla |\psi||^2 \,d\sx < \infty$,
		and hence $|\psi| \in H^1(\R^{2N} \setminus \bDelta) = H^1(\R^{2N})$ 
		by Lemma \ref{lem:folklore}.
	\end{proof}
	
	\begin{thm} \label{thm:anyon_domains}
		For arbitrary $\alpha \in \R$ we have that
		$D^{\textup{min}} = D^{\textup{max}}$ and hence
		$\hat{T}_{\uA} = \frac{1}{2} (D^{\textup{max}})^* D^{\textup{max}}$ 
		is equal to the Friedrichs extension
		$\hat{T}_{\uA} = \frac{1}{2} (D^{\textup{min}})^* D^{\textup{min}}$.
	\end{thm}
	\begin{proof}
		Take $\psi \in \Dc_{\uA}^\alpha$.
		We need to show that $\psi$ 
		can be approximated by a sequence 
		$\psi_n \in C_c^\infty(\R^{2N} \setminus \bDelta) \cap L^2_{\textup{sym}}$
		s.t. $\|\psi_n - \psi\| \to 0$ and 
		$\|D(\psi_n - \psi)\| \to 0$ as $n \to \infty$.
		
		Step 1: We start by proving that we without loss of generality
		can assume $\psi \in L^2_{c,\textup{sym}}(\R^{2N})$,
		by finding a sequence $\psi_n \in L^2_{c,\textup{sym}}(\R^{2N})$
		s.t. $D\psi_n \in L^2(\R^{2N})$, $\|\psi_n-\psi\| \to 0$ and $\|D(\psi_n - \psi)\| \to 0$.
		We take $\psi_n := \chi_n\psi$ where $\chi_n(\sx) := \chi(\sx/n)$
		and $\chi \in C^\infty_c(\R^{2N};[0,1]) \cap L^2_\textup{sym}$ 
		with $\chi \equiv 1$ for $|\sx|<1$ and 0 for $|\sx|>2$.
		Then $\psi_n$ is symmetric with $\supp \psi_n \subseteq B_{2n}(0)$,
		$$
			\|\psi_n-\psi\|^2 \le \int_{B_n(0)^c} |\psi|^2 \,d\sx \to 0,
			\quad n \to \infty,
		$$
		as well as
		$D\psi_n = -\frac{i}{n}\nabla \chi(\sx/n) \psi + \chi_n D\psi \in L^2(\R^{2N})$,
		and
		$$
			\|D\psi_n - D\psi\| 
			\le \frac{1}{n} \|\nabla\chi(\sx/n) \psi\| + \|(\chi_n -1)D\psi\|
			\le \frac{C}{n} \|\psi\| 
			+ \|D\psi\|_{L^2(B_n^c)} \to 0.
		$$
		
		Step 2: We next show that it is sufficient to 
		assume $\psi \in L^\infty_{c,\textup{sym}}(\R^{2N})$.
		By step 1, we have $\psi \in L^2_{c,\textup{sym}}(\R^{2N})$ 
		with $D\psi \in L^2_c(\R^{2N})$.
		Let $\psi_n(\sx) := \chi_n(|\psi|(\sx))\psi(\sx)$
		where $\chi_n$ is the cut-off function 
		defined analogously as above but on $\R$.
		Then $|\psi_n(\sx)| \le 2n$ and $\psi_n \in L^\infty_c(\R^{2N})$.
		Furthermore, on $\R^{2N} \setminus \bDelta$ we have
		$\psi \in H^1_\textup{loc}(\R^{2N} \setminus \bDelta)$ 
		and, by Theorem 6.16 in \cite{Lieb-Loss:01}
		(cp. also Theorem 7.22 in \cite{Lieb-Loss:01}),
		$\psi_n \in H^1_\textup{loc}(\R^{2N} \setminus \bDelta)$ and
		$$
			\nabla \psi_n = \chi_n(|\psi|)\nabla\psi + \psi\chi_n'(|\psi|)\nabla|\psi| 
			\ \in L^2_\textup{loc}(\R^{2N} \setminus \bDelta).
		$$
		Hence,
		$$
			D\psi_n = \underbrace{\chi_n(|\psi|)}_{\in L^\infty} D\psi 
			- i\underbrace{\psi\chi_n'(|\psi|)}_{\in L^\infty} \nabla|\psi| 
			\ \in L^2_\textup{loc}(\R^{2N} \setminus \bDelta),
		$$
		with the full $L^2$-norm of $D\psi_n$ on $\R^{2N}$ then bounded by
		$$
			\|\chi_n(|\psi|)\|_\infty \|D\psi\|_2
			+ \|\psi \chi_n'(|\psi|)\|_\infty \|\nabla|\psi|\|_2 
			\le (1+2n)\|D\psi\|_2 < \infty,
		$$
		where we have used the diamagnetic inequality, 
		Lemma \ref{lem:diamagnetic_inequality}.
		Hence, $\psi_n \in \Dc_{\uA}^\alpha$ and furthermore
		$$
			\|\psi_n-\psi\|^2 \le \int_{|\psi|>n} |\psi|^2 \,d\sx \to 0,
		$$
		and
		$$
			\|D\psi_n - D\psi\| 
			\le \|(\chi_n(|\psi|)-1) D\psi\| + \|\psi\chi_n'(|\psi|) \nabla\psi\|,
		$$
		where the first term tends to zero by dominated convergence
		and for the second we have
		$|\psi|\chi_n'(|\psi|) \le C \chi_{\{|\psi|>n\}}$, hence also
		$\int_{\R^{2N}} |\psi \chi_n'(|\psi|)|^2 \big|\nabla|\psi|\big|^2 \,dx \!\to 0$.
		
		Step 3: Our next step is to show that we can assume
		$\psi \in L^\infty_{c,\textup{sym}}(\R^{2N} \!\setminus \bDelta)$, i.e.
		$\psi$ supported away from diagonals.
		By steps 1-2 we have $\psi \in L^\infty_{c,\textup{sym}}(\R^{2N})$
		and hence $D\psi \in L^2_c(\R^{2N})$.
		As in the proof of Lemma \ref{lem:folklore} we define
		$\psi_\ep(\sx) := \prod_{j<k} \varphi_\ep(\bx_j - \bx_k)\psi(\sx)$
		which is in $L^\infty_{c,\textup{sym}}(\R^{2N} \setminus \bDelta)$
		for $\ep>0$. We obtain
		$$
			D\psi_\ep = {\textstyle\prod} \varphi_\ep D\psi 
			-i(\nabla {\textstyle\prod} \varphi_\ep) \psi
			\in L^2_c(\R^{2N} \setminus \bDelta),
		$$
		and
		$$
			\|D\psi_n - D\psi\| 
			\le \|({\scriptstyle\prod} \varphi_\ep - 1)D\psi\| 
			+ \|(\nabla {\scriptstyle\prod} \varphi_\ep)\psi\|,
		$$
		where both terms tend to zero as $\ep \to 0$
		as in Lemma \ref{lem:folklore}.
		
		Step 4: By steps 1-3 we can assume 
		$\psi \in L^\infty_{c,\textup{sym}}(\R^{2N} \setminus \bDelta)$.
		Hence $\psi$ is compactly supported on some open set
		$\Omega := (\R^{2N} \setminus \bDelta) \cap B_R(0)$, 
		and also $D\psi = -i\nabla\psi + A\psi \in L^2_c(\Omega)$. 
		We have then since $A\psi \in L^2_c(\Omega)$ that
		$\psi \in H^1_0(\Omega)$, so there is a sequence 
		$\psi_n \in C^\infty_c(\Omega)$ s.t.
		$$
			\|\psi_n - \psi\|_{L^2(\Omega)} \to 0, \quad \text{and} \quad
			\|\nabla\psi_n - \nabla\psi\|_{L^2(\Omega)} \to 0,
		$$
		as $n \to \infty$.
		Hence $D\psi_n \in L^2(\R^{2N})$ and
		\begin{multline*}
			\|D(\psi_n-\psi)\|_{L^2(\R^{2N})}
			= \|{-i\nabla(\psi_n-\psi)} + A(\psi_n-\psi)\|_{L^2(\Omega)} \\
			\le \|\nabla\psi_n - \nabla\psi\|_{L^2(\Omega)}
			+ \|A\|_{L^\infty(\Omega)} \|\psi_n - \psi\|_{L^2(\Omega)}
			\to 0,
		\end{multline*}
		which proves the theorem.
	\end{proof}

\section{Local exclusion} \label{sec:local_exclusion}

	Since our wave functions are modeled using bosonic symmetry,
	but with the non-trivial exchange statistics represented by 
	an interaction,
	we cannot take advantage of the usual exclusion principle
	encoded in the total (global) antisymmetry of the wave function.
	Instead we recall the following \emph{local} consequence of 
	the Pauli principle for fermions,
	given as Lemma 5 in \cite{Dyson-Lenard:67},
	and which was used by Dyson and Lenard in their 
	proof of stability of matter in the bulk
	(see also \cite{Dyson:68,Lenard:73}).
	
	\begin{lem}[Local exclusion for fermions in 3D] 
		\label{lem:local_exclusion_F}
		Let $\psi \in \bigwedge^n L^2(\R^3)$ be a wave function of $n$ 
		fermions in $\R^3$ and let $\Omega$ be a ball of radius $\ell$.
		Then
		\begin{equation} \label{local_exclusion_F}
			\int_{\Omega^n} \sum_{j=1}^n |\nabla_j \psi|^2 \,d\sx
			\ \ge \ (n-1) \frac{\xi^2}{\ell^2} 
				\int_{\Omega^n} |\psi|^2 \,d\sx,
		\end{equation}
		where $\xi \approx 2.082$ 
		is the smallest positive root of the equation
		\begin{equation}
			\frac{d^2}{dx^2} \frac{\sin x}{x} = 0.
		\end{equation}
	\end{lem}

	We refer to such local energy bounds as a
	\emph{local exclusion principle} since it implies that the energy is
	strictly positive whenever we have more than one particle,
	and hence that the particles cannot occupy the same single-particle state 
	(which on a local region with free boundary conditions 
	would be the zero-energy ground state).
	The inequality follows by expanding $\psi$ in the eigenfunctions of the 
	Neumann Laplacian on $\Omega$, 
	and it turns out to be sufficient with such a weak \emph{linear} 
	dependence on $n$.
	A corresponding family of bounds of this form 
	for generalized exchange statistics 
	forms the starting point for our proofs of 
	kinetic energy and Lieb-Thirring inequalities
	for such statistics.

\subsection{Local exclusion for anyons and bounds for the energy of the ideal anyon gas}

	The following 
	local exclusion principle for anyons
	was proved in \cite{Lundholm-Solovej:anyon}, starting from 
	a local pairwise relative magnetic Hardy inequality.
	The method then used for lifting such a pairwise energy bound to
	a bound for the full kinetic energy on a local region
	will be illustrated below for the case of 
	intermediate statistics in 1D
	(and is also in the case of fermions in three dimensions 
	applicable for an alternative proof of 
	Lemma \ref{lem:local_exclusion_F} above, 
	up to the value of the constant).

	\begin{lem}[Local exclusion for anyons] \label{lem:local_exclusion_A}
		Let $\psi \in \Dc_{\uA}^\alpha$ be a wave function of $n$ anyons 
		and let $\Omega \subseteq \R^2$ be either a disk or a square,
		with area $|\Omega|$.
		Then
		\begin{equation} \label{local_exclusion_A}
			\int_{\Omega^n} \sum_{j=1}^n |D_j \psi|^2 \,d\sx
			\ \ge \ (n-1) \frac{c_\Omega C_{\alpha,n}^2}{|\Omega|} 
				\int_{\Omega^n} |\psi|^2 \,d\sx,
		\end{equation}
		where $c_\Omega$ is a constant which satisfies $c_\Omega \ge 0.169$
		for the disk and $c_\Omega \ge 0.112$ for the square,
		and $C_{\alpha,n}$ is defined in \eqref{C_alpha}.
		Hence, defining the 
		local kinetic energy on $\Omega$ for a normalized 
		$N$-particle wave function
		$\psi \in \Dc_{\uA}^\alpha$
		\begin{equation} \label{local_kinetic_energy_A}
			T_{\uA}^\Omega := \sum_{j=1}^N \int_{\R^{2N}} \frac{1}{2} |D_j \psi|^2 \,\chi_\Omega(\bx_j) \,d\sx,
		\end{equation}
		we obtain the following local energy bound
		in terms of $\rho$
		\begin{equation} \label{local_exclusion_density_A}
			T_{\uA}^\Omega \ \ge \ \frac{c_\Omega C_{\alpha,N}^2}{2|\Omega|} 
					\left( \int_\Omega \rho(\bx) d\bx \ - 1 \right)_+.
		\end{equation}
	\end{lem}

	We shall consider an immediate application,
	partially along the lines of \cite{Dyson-Lenard:67}, 
	of this local exclusion principle for anyons 
	in the form \eqref{local_exclusion_A}
	to an explicit lower bound for the ground state energy
	of a gas of non-interacting anyons.
	The numerical constant we obtain 
	by this comparatively simple method is much better
	than the one following from the Lieb-Thirring inequality 
	\eqref{kinetic-LT-anyons} for anyons
	given below and proven in \cite{Lundholm-Solovej:anyon}.

	\begin{thm}[Ground state energy for $N$ anyons in a box] \label{thm:anyon_gas}
		Let $\psi \in \Dc_{\uA}^\alpha$  be a normalized $N$-anyon wave function
		supported on a square $Q_L \subseteq \R^2$ of side length $L$. Then
		\begin{equation} \label{anyon_kinetic_bound}
			T_{\uA} =
			\frac{1}{2} \int_{\R^{2N}} \sum_{j=1}^N |D_j \psi|^2 \,d\sx 
			\ \ge \ \frac{c_\Omega}{2} c_N C_{\alpha,N}^2 \frac{N^2}{L^2},
		\end{equation}
		where $c_\Omega$ is is the constant in Lemma \ref{lem:local_exclusion_A} 
		for the disk, and
		$$
			c_N := \sup_{\gamma>0} \left( \frac{1}{\pi\gamma^2} - \frac{(1+2\gamma/\sqrt{N})^2}{\pi^2\gamma^4} \right).
		$$
		Hence, 
		the energy per unit area of the ideal anyon gas is bounded below by
		$$
			\frac{T_{\uA}}{L^2} 
			\ \ge \ \frac{c_\Omega}{2} c_N C_{\alpha,N}^2 \bar{\rho}^2
			\stackrel{N > 10^6}{\ge} 
			0.021 \,C_{\alpha,N}^2 \,\bar{\rho}^2, 
		$$
		which,
		whenever the statistics parameter 
		$\alpha = \frac{\mu}{\nu}$
		is an odd numerator reduced fraction,
		remains finite in the thermodynamic limit and results in the bound
		\begin{equation} \label{anyon_gas_bound}
			\frac{T_{\uA}}{L^2} \ge 0.021 \frac{\bar{\rho}^2}{\nu^2},
		\end{equation}
		as $N \to \infty$ and $L \to \infty$ 
		while the total density $\bar{\rho} := N/L^2$ is kept fixed.
	\end{thm}
	\begin{proof}
		Our approach is very similar to
		the proof of Theorem 8 in \cite{Dyson-Lenard:67},
		which used the local exclusion principle 
		\eqref{local_exclusion_F} for fermions.
		Introduce an arbitrary length $\ell > 0$ and write
		$Q_{L,\ell} := Q_L + B_\ell(0)$.
		Then, 
		$$
			2T_{\uA} 
			= \int_{\R^{2N}} \sum_{j=1}^N |D_j \psi|^2 \frac{1}{\pi\ell^2} 
				\int_{Q_{L,\ell}} \chi_{B_\ell(\by)}(\bx_j) \,d\by \,d\sx,
		$$ 
		and after changing the order of integration and
		inserting the partition of unity
		\begin{multline} \label{partition_of_unity_B}
			1 = \prod_{k=1}^N \left( \chi_{B_\ell(\by)}(\bx_k) + \chi_{B_\ell(\by)^c}(\bx_k) \right) \\
			= \sum_{A \subseteq \{1,\ldots,N\}} \prod_{k \in A} 
				\chi_{B_\ell(\by)}(\bx_k) \prod_{k \notin A} \chi_{B_\ell(\by)^c}(\bx_k),
		\end{multline}
		we find
		\begin{multline*}
			2T_{\uA} 
			= \frac{1}{\pi\ell^2} \int_{Q_{L,\ell}} \sum_{A \subseteq \{1,\ldots,N\}}
				\int_{\R^{2N}} \sum_{j=1}^N |D_j \psi|^2 \,\chi_{B_\ell(\by)}(\bx_j) \\
				\times \prod_{k \in A} \chi_{B_\ell(\by)}(\bx_k) 
				\prod_{k \notin A} \chi_{B_\ell(\by)^c}(\bx_k) \,d\sx \,d\by \\
			= \frac{1}{\pi\ell^2} \int_{Q_{L,\ell}} \sum_{A \subseteq \{1,\ldots,N\}}
				\int\limits_{(B_\ell(\by)^c)^{N-|A|}}
				\int\limits_{(B_\ell(\by))^{|A|}}
				\sum_{j \in A} |D_j \psi|^2 
				\prod_{k \in A} d\bx_k \prod_{k \notin A} d\bx_k \,d\by.
		\end{multline*}
		We now apply \eqref{local_exclusion_A} 
		to each term 
		in the first summation above, which involves
		a partition $A$ of the $N$ particles into 
		$n := |A|$ of them being inside the disk $B_\ell(\by)$, 
		while the remaining $N-n$ residing outside, 
		and therefore whose contributions to the magnetic potentials 
		$\bA_{j \in A}$
		can be gauged away. 
		Thus, we find
		\begin{multline*}
			2T_{\uA} \ge \frac{c_\Omega}{\pi^2\ell^4} 
				\int\limits_{Q_{L,\ell}} \sum_{A \subseteq \{1,\ldots,N\}}
				C_{\alpha,|A|}^2 \left( |A| - 1 \right)_+ \\ \times 
				\int\limits_{(B_\ell(\by)^c)^{N-|A|}}
				\int\limits_{(B_\ell(\by))^{|A|}}
				|\psi|^2 \prod_{k \in A} d\bx_k \prod_{k \notin A} d\bx_k \,d\by \\
			\ge \frac{c_\Omega C_{\alpha,N}^2}{\pi^2\ell^4} \int_{Q_{L,\ell}} 
				\sum_{A \subseteq \{1,\ldots,N\}} \!\!\!\!|A|
				\int_{\R^{2N}} |\psi|^2
				\prod_{k \in A} \chi_{B_\ell(\by)}(\bx_k) 
				\prod_{k \notin A} \chi_{B_\ell(\by)^c}(\bx_k) \,d\sx \,d\by \\
			\quad -\ \frac{c_\Omega C_{\alpha,N}^2}{\pi^2\ell^4} \int_{Q_{L,\ell}} 
				\int_{\R^{2N}} |\psi|^2 \sum_{A \subseteq \{1,\ldots,N\}}
				\prod_{k \in A} \chi_{B_\ell(\by)}(\bx_k) 
				\prod_{k \notin A} \chi_{B_\ell(\by)^c}(\bx_k) \,d\sx \,d\by
		\end{multline*}
		We then revert the above procedure using 
		\eqref{partition_of_unity_B} and
		\begin{multline*}
			\sum_{A \subseteq \{1,\ldots,N\}} \underbrace{|A|}_{=\sum_{j \in A}}
				\prod_{k \in A} \chi_{B_\ell(\by)}(\bx_k) 
				\prod_{k \notin A} \chi_{B_\ell(\by)^c}(\bx_k) \\
			= \sum_{j=1}^N \sum_{A \subseteq \{1,\ldots,N\}} \chi_{B_\ell(\by)}(\bx_j)
				\prod_{k \in A} \chi_{B_\ell(\by)}(\bx_k) 
				\prod_{k \notin A} \chi_{B_\ell(\by)^c}(\bx_k) \\
			= \sum_{j=1}^N \chi_{B_\ell(\by)}(\bx_j),
		\end{multline*}
		and hence 
		\begin{multline*}
			2T_{\uA} \ge \frac{c_\Omega C_{\alpha,N}^2}{\pi^2\ell^4} 
				\int_{\R^{2N}} |\psi|^2 
				\int_{Q_{L,\ell}} \left( 
					\sum_{j=1}^N \chi_{B_\ell(\by)}(\bx_j) - 1
					\right) d\by \,d\sx \\
			\ge \frac{c_\Omega C_{\alpha,N}^2}{\pi^2\ell^4} \left( 
				N\pi\ell^2 - (L+2\ell)^2 \right) 
			= c_\Omega C_{\alpha,N}^2 g(\gamma) \frac{N^2}{L^2},
		\end{multline*}
		setting $\ell := \frac{\gamma L}{\sqrt{N}}$ 
		with $\gamma > 0$, and
		$$
			g(\gamma) := \frac{1}{\pi\gamma^2} - \frac{(1+2\gamma/\sqrt{N})^2}{\pi^2\gamma^4}.
		$$
		This function has a positive maximum $c_N \ge 0.0023$ for $N \ge 2$,
		while $c_N \to 1/4$ as $N \to \infty$.
		In the case $N=1$ we are 
		left with the trivial bound on the energy.
	\end{proof}

	Note that for small $N$ it might be preferable to 
	use the Dirichlet bound for bosons instead
	(from the diamagnetic inequality \eqref{diamag_pointwise}):
	$$
		\sum_j \int_{Q_L^N} |D_j \psi|^2 \,d\sx 
		\ge \sum_j \int_{Q_L^N} |\nabla_j |\psi||^2 \,d\sx
		\ge 2N\frac{\pi^2}{L^2} \int_{Q_L^N} |\psi|^2 \,d\sx.
	$$
	
	This resulting rough --- but nevertheless non-trivial --- 
	bound \eqref{anyon_gas_bound} for the ideal anyon gas
	should perhaps be compared with the 
	`correct' ground state energy per unit volume
	for a non-interacting gas of density $\bar{\rho}$ of 
	spinless fermions in two dimensions,
	which in our conventions is
	\begin{equation} \label{fermion_gas_energy}
		\lim_{N,L \to \infty} \frac{T_0}{L^2} =
		\pi \bar{\rho}^2. 
	\end{equation}
	Suggestions for improving the constant $c_\Omega$
	were given in \cite{Lundholm-Solovej:anyon}.
	Also note for comparison that the bound 
	\eqref{local_exclusion_A} for the disk holds with
	$c_\Omega = \pi\xi'^2 \approx 10.65$ in the case that $\alpha = 1$
	(cp. with Lemma \ref{lem:local_exclusion_F}), 
	where $\xi'$ denotes the first zero of the derivative
	of the Bessel function $J_1$.

\subsection{Local exclusion in one dimension}

	As for the anyons \cite{Lundholm-Solovej:anyon} we 
	consider the local increase in energy due to the statistical interaction,
	starting from a pairwise relative consideration.

	\begin{lem}[Pairwise relative LL] \label{lem:relative_S}
		For $\eta \ge 0$, consider the operator:
		$$
			H_\eta = -\frac{d^2}{dr^2} + 2\eta \delta_0(r)
		$$
		on a symmetric interval $[-l,l]$ 
		with Neumann boundary conditions at $r=\pm l$.
		Then $H_\eta \ge \xi_{\uS}(\eta l)^2/l^2$,
		where $\xi_{\uS}(y)$ is defined as the smallest 
		nonnegative solution to
		$\xi \tan \xi = y$ given $y \ge 0$.
	\end{lem}

	\begin{proof}
		Note that for $\eta\ge 0$ there are only non-negative eigenvalues
		$\lambda$ of $H_\eta$.
		The general solution to the corresponding boundary value problem
		\begin{eqnarray*}
			-u'' &=& \lambda u, \quad \text{on $(-l,0) \cup (0,l)$}, \\
			u'(0^+) = -u'(0^-) &=& \eta u(0), \\
			u'(l) = u'(-l) &=& 0,
		\end{eqnarray*}
		with $u(r) = u(-r)$ and $\lambda \ge 0$,
		is $u(r) = C\left( \cos \sqrt{\lambda}r 
			+ \frac{\eta}{\sqrt{\lambda}} \sgn r \sin \sqrt{\lambda}r \right)$.
		The condition at $r=\pm l$ demands 
		$\sqrt{\lambda} \tan \sqrt{\lambda}l = \eta$,
		i.e. $\lambda = \xi(\eta l)^2/l^2$
		is the smallest eigenvalue.
	\end{proof}

	\begin{figure}[t]
		\centering
		\includegraphics[scale=0.29]{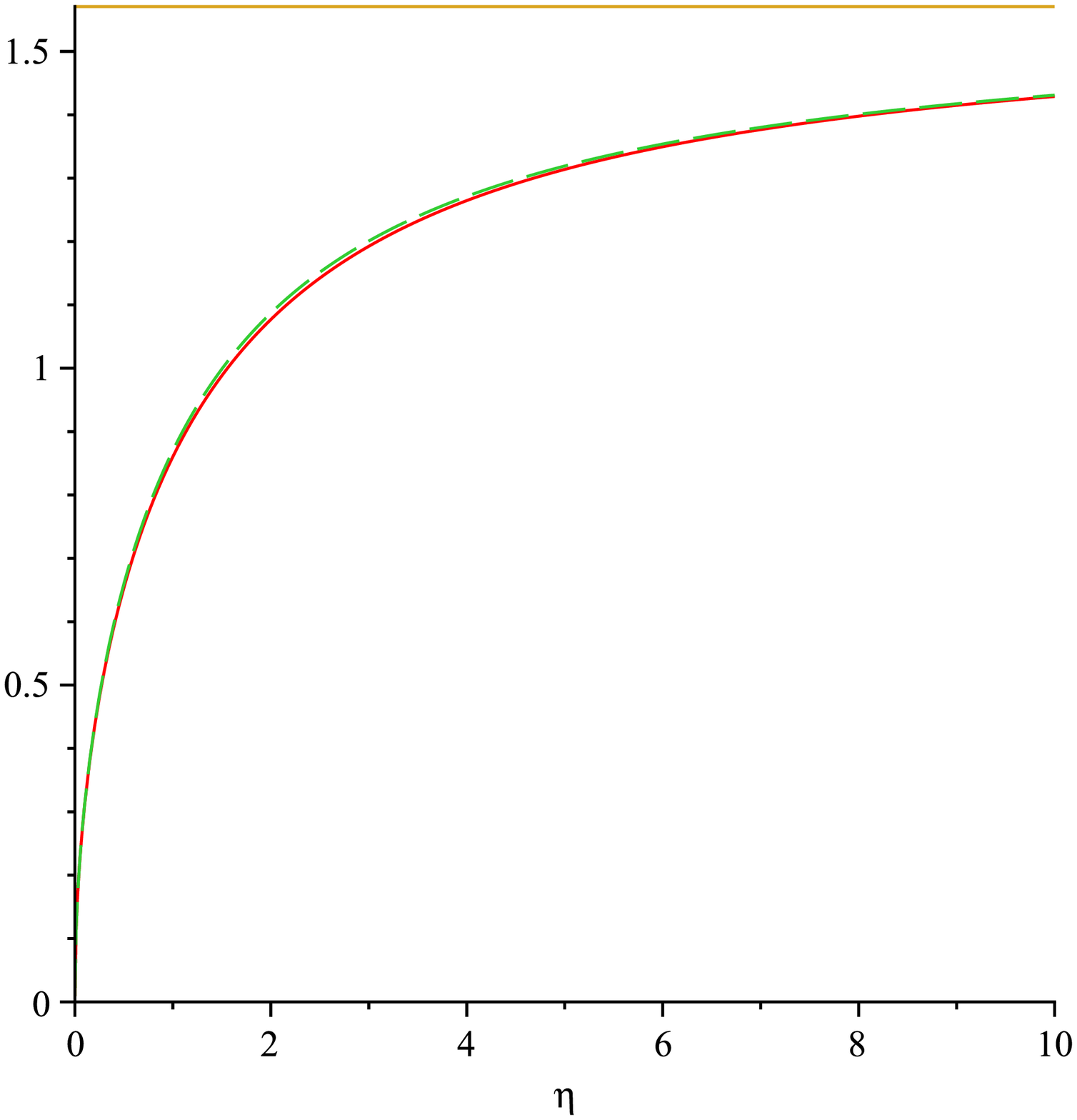}
		\includegraphics[scale=0.29]{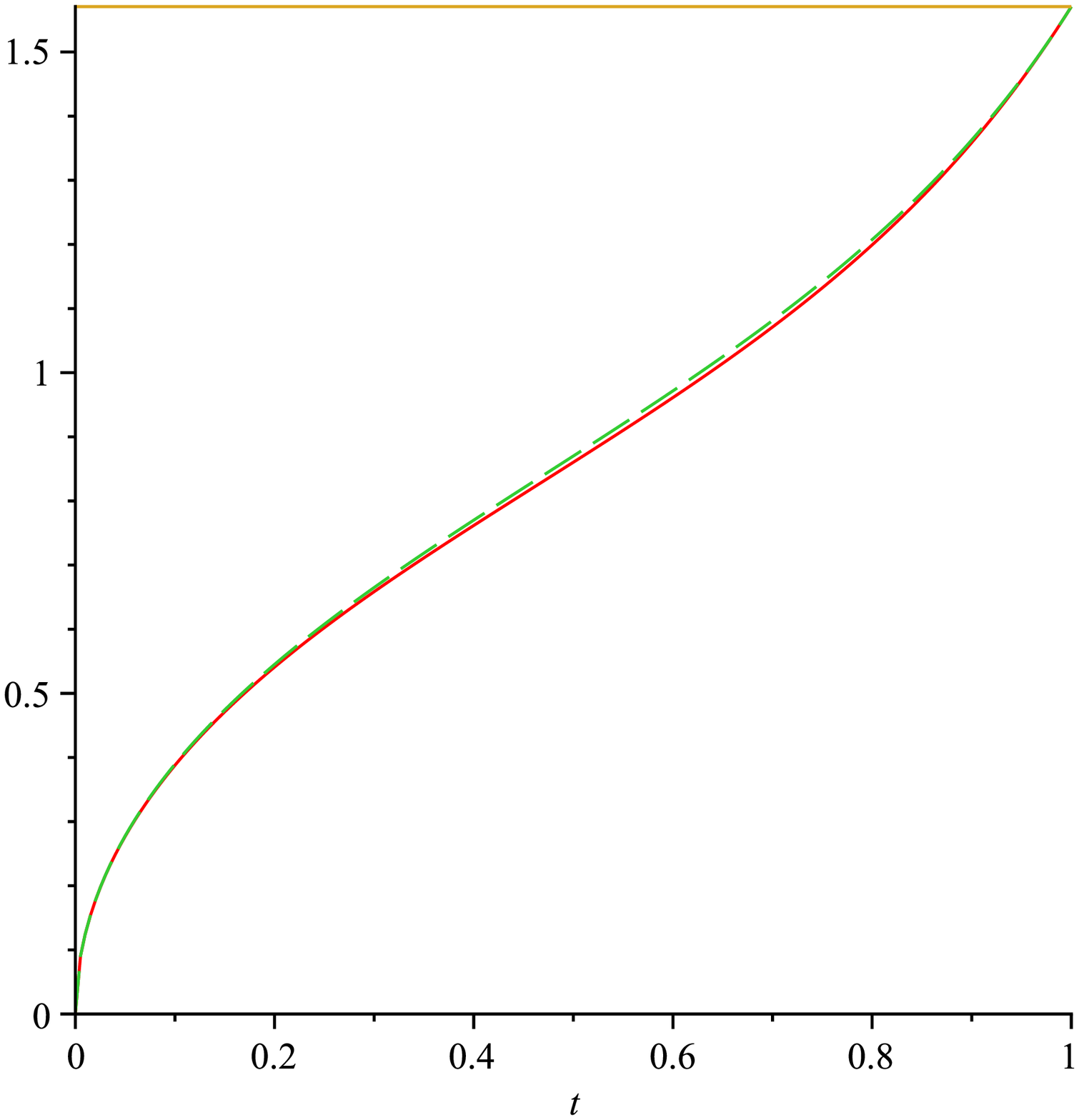}
		\caption{Plot of $\xi_{\uS}(\eta)$ (red solid) and 
		$\arctan \sqrt{\eta + 4\eta^2/\pi^2}$ (green dashed) 
		as a function of $\eta$ resp. $t=\frac{2}{\pi}\arctan(\eta)$.}
		\label{fig:xiS}
	\end{figure}
	
	Note that $\xi_{\uS}$ is monotone 
	and $\xi_{\uS}(y) \sim \sqrt{y}$ for small $y$,
	while $\xi_{\uS}(y) \to \frac{\pi}{2}$ as $y \to \infty$.
	A numerically very good approximation to $\xi_{\uS}$ for all $y$ is
	(see Figure \ref{fig:xiS})
	\begin{equation} \label{approx_S}
		\xi_{\uS}(y) \approx \arctan \sqrt{y + \frac{4}{\pi^2}y^2}.
	\end{equation}
	We shall also find use for the following additional observations:
	\begin{lem} \label{lem:Schroedinger_properties}
		The map $y \mapsto \xi_{\uS}(y)^2$, $y \ge 0$, is monotone and concave,
		so that for $\eta \ge 0$
		\begin{equation} \label{Schroedinger_concavity}
		\begin{aligned}
			\xi_{\uS}(\eta/x)^2 x &\ge \xi_{\uS}(\eta)^2, \quad x \ge 1, \\
			\xi_{\uS}(\eta/x)^2 x &\le \xi_{\uS}(\eta)^2, \quad 0 \le x \le 1.
		\end{aligned}
		\end{equation}
		Furthermore, the map 
		$x \mapsto \xi_{\uS}(\eta/x)^2 x^3$, $x \ge 0$,
		is monotone and convex for any $\eta \ge 0$.
	\end{lem}
	\begin{proof}
		Differentiating the identity $\xi\tan\xi = y$ 
		w.r.t. $y>0$
		we obtain
		$\xi' (\xi^2 + y(y+1)) = \xi$, 
		hence $\xi'(y) > 0$. 
		Differentiating once more gives
		$$
			\xi''(\xi^2 + y(y+1)) = -2\xi'(\xi\xi' + y),
		$$
		so
		$$
			(\xi^2)'' = 2\xi'^2 + 2\xi\xi'' 
			= - \frac{4\xi^2}{\xi^2 + y(y+1)} \left(
				\frac{\xi^2}{\xi^2 + y(y+1)} + y - \frac{1}{2}
			\right) < 0,
		$$
		since the expression in brackets is positive for $y>0$. 
		The estimates 
		\eqref{Schroedinger_concavity} then follow since 
		$\xi(0)=0$.
		The proof of convexity for $\xi(\eta/x)^2 x^3$
		follows by similar analysis.
	\end{proof}

	\begin{lem}[Pairwise relative CS] \label{lem:relative_H}
		For $\alpha \ge 0$, consider nonnegative realizations of the operator:
		$$
			H_\alpha = -\frac{d^2}{dr^2} + \frac{\alpha(\alpha-1)}{r^2}
		$$
		on a symmetric interval $[-l,l]$ 
		with Neumann boundary conditions at $r=\pm l$.
		(In the case $-1/2 < \alpha < 3/2$ 
		one has to pick the correct self-adjoint extension,
		matching the b.c. \eqref{boundary_cond_H}
		--- see below, and 
		Section \ref{sec:identical_1D}.)
		Then $H_\alpha \ge \xi_{\uH}(\alpha)^2/l^2$,
		where $\xi_{\uH}(\alpha)$ is defined as the smallest 
		nonnegative solution to
		$$
			J(\xi) + 2\xi J'(\xi) = 0,
		$$
		where $J = J_{\alpha - \frac{1}{2}}$ is the Bessel function
		of order $\alpha - \frac{1}{2}$.
	\end{lem}

	\begin{proof}
		The general solution to the corresponding boundary value problem
		on $[0,l]$:
		\begin{eqnarray*}
			-u'' + \frac{\alpha(\alpha-1)}{r^2}u &=& \lambda u, \quad \text{on $(0,l)$}, \\
			u'(l) &=& 0,
		\end{eqnarray*}
		is via the ansatz $u := r^{\frac{1}{2}}v$ found to be
		(we have assumed $\lambda \ge 0$ 
		and focus on the non-degenerate case $\alpha \neq 1/2$)
		$$
			u(r) = C_1 r^{\frac{1}{2}} J_{\alpha - \frac{1}{2}}(\sqrt{\lambda} r)
			+ C_2 r^{\frac{1}{2}} J_{\frac{1}{2} - \alpha}(\sqrt{\lambda} r).
		$$
		For $\alpha \ge 3/2$, $C_2=0$ is enforced by 
		the requirement of square-integrability,
		while for $\alpha < 3/2$ this choice is taken
		to define the domain of $H_\alpha$.
		Note that for $\alpha > 1/2$ the choice $C_2=0$ is required in order for 
		$r^\alpha (r^{-\alpha}u)' \in L^2([0,l])$
		and hence for (a corresponding two-particle version of) $u$
		to be in the domain of $Q_\alpha^{\textup{max}}$.
		The Neumann condition at $r=l$ then demands 
		$J_{\alpha-\frac{1}{2}}(\sqrt{\lambda}l) + 2\sqrt{\lambda}l J_{\alpha-\frac{1}{2}}'(\sqrt{\lambda}l) = 0$,
		i.e. $\lambda = \xi(\alpha)^2/l^2$
		is the smallest eigenvalue.
	\end{proof}

	\begin{figure}[t]
		\centering
		\includegraphics[scale=0.4]{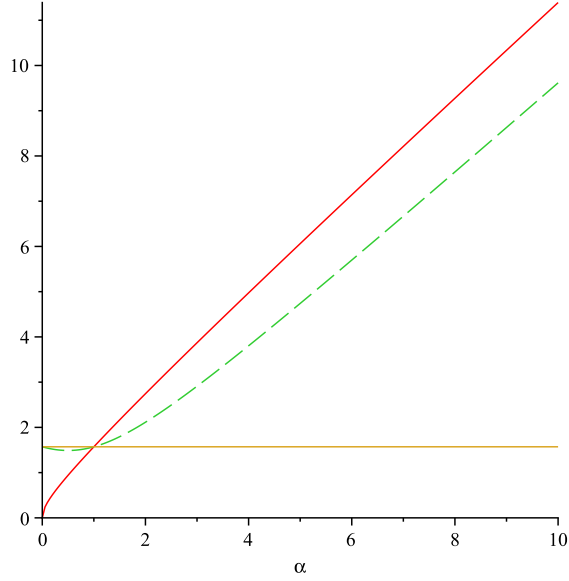}
		\caption{Plot of $\xi_{\uH}(\alpha)$ (red solid) and 
		$\sqrt{\pi^2/4+\alpha(\alpha-1)}$ (green dashed) 
		as a function of $\alpha \ge 0$.}
		\label{fig:xiH}
	\end{figure}
	
	Note that according to the above $\xi_{\uH}(0) = 0$,
	while $\xi_{\uH}(1) = \frac{\pi}{2}$.	
	A numerically very good approximation to $\xi_{\uH}$ on the 
	interval $[0,1]$ is
	\begin{equation} \label{approx_H_small}
		\xi_{\uH}(\alpha) \approx \sqrt{\alpha + \left(\frac{\pi^2}{4} - 1\right)\alpha^2},
		\qquad 0 \le \alpha \le 1.
	\end{equation}
	However, note that $\xi_{\uH}$ continues to grow
	approximately like $\xi_{\uH}(\alpha) \sim \alpha$ 
	for $\alpha > 1$
	(see Figure \ref{fig:xiH}).
	In any case we have 
	from the lowest eigenvalue of
	the Laplacian with Dirichlet boundary condition at $r=0$,
	Neumann at $r=l$,
	and with a potential $\ge \alpha(\alpha-1)/l^2$:
	\begin{equation} \label{approx_H_large}
		\xi_{\uH}(\alpha) \ge \sqrt{ \frac{\pi^2}{4} + \alpha(\alpha-1)},
		\qquad \text{for $\alpha \ge 1$}.
	\end{equation}

	Now, consider the Neumann kinetic energy for $n \ge 2$ particles 
	on an interval $Q = [a,b]$. Using the identity
	\begin{equation} \label{summation_identity}
		n \sum_{j=1}^n |\partial_j \psi|^2 = 
		\sum_{1 \le j < k \le n} |\partial_j \psi - \partial_k \psi|^2
		+ \left| \sum_{j=1}^n \partial_j \psi \right|^2
	\end{equation}
	to separate out the center-of-mass motion, we find
	\begin{multline} \label{relative_kinetic_bound}
		\!\!\!\!\!\!\!
		\int_{Q^n} T_{\uS/\uH}^n(\psi;\sx) \,d\sx 
		\ge \int_{Q^n} \left( 
			\frac{1}{2n} \sum_{j < k} |\partial_j \psi - \partial_k \psi|^2
			+ \sum_{j<k} V_{\uS/\uH}(x_j - x_k) |\psi|^2 \right) d\sx \\
		\ge \frac{2}{n} \sum_{j<k}
		\int_{Q^{n-2}} \int_Q \int_{[-\delta(R),\delta(R)]}
			\left( |\partial_r \psi|^2 + V_{\uS/\uH}(r) |\psi|^2 \right) dr \,dR \,d\sx' \\
		\ge \frac{2}{n} \sum_{j<k}
		\int_{Q^{n-2}} \int_Q \frac{\xi_{\uS/\uH}^2}{\delta(R)^2}
			\int_{[-\delta(R),\delta(R)]} |\psi|^2 \,dr \,dR \,d\sx',
	\end{multline}
	where we for each pair $(j,k)$ 
	split the parameterization of $Q^n$ into $n-2$ variables 
	$\sx' = (x_1,\ldots,x\!\!\!\!\diagup\!_j,\ldots,x\!\!\!\!\diagup\!_k,\ldots,x_N)$,
	a pairwise center-of-mass coordinate $R := (x_j+x_k)/2 \in Q$,
	and a relative coordinate 
	$r := (x_j - x_k) \in [-\delta(R),\delta(R)]$. 
	We have set $\delta(R) := 2\min \{|R-a|,|R-b|\}$,
	and in the above we also used that $V_{\uS/\uH} \ge 0$
	followed by Lemma \ref{lem:relative_S} or \ref{lem:relative_H}.

	In the Calogero-Sutherland case $\xi_{\uH} = \xi_{\uH}(\alpha)$ is constant
	on the interval $Q$ and 
	a local exclusion principle follows
	in a similar way as for the anyons:
	
	\begin{lem}[Local exclusion for CS] \label{lem:local_exclusion_H}
		We have for $\alpha \ge 1$
		\begin{equation} \label{local_exclusion_H}
			\int_{Q^n} T_{\uH}^n(\psi;\sx) \,d\sx
			\ \ge \ (n-1) \frac{\xi_{\uH}(\alpha)^2}{|Q|^2} 
				\int_{Q^n} |\psi|^2 \,d\sx,
		\end{equation}
		and hence, defining the local kinetic energy on $Q$
		for a normalized $N$-particle wave function 
		$\psi \in \Dc_{\uH}^\alpha$
		\begin{equation} \label{local_kinetic_energy_H}
			T_{\uH}^Q := \sum_{j=1}^N \int_{\R^N} \frac{1}{2} \left( |\partial_j \psi|^2 
				+ \sum_{(j\neq) k=1}^N V_{\uH}(x_j - x_k) |\psi|^2 \right) \chi_Q(x_j) \,d\sx,
		\end{equation}
		we obtain the local energy bound
		\begin{equation} \label{local_exclusion_density_H}
			T_{\uH}^Q \ \ge \ \frac{\xi_{\uH}(\alpha)^2}{|Q|^2} 
				\left( \int_Q \rho(x) dx \ - 1 \right)_+.
		\end{equation}
	\end{lem}
	\begin{proof}
		We use
		$\delta(R)^{-2} \ge |Q|^{-2}$
		in \eqref{relative_kinetic_bound} 
		to obtain \eqref{local_exclusion_H}.
		Then we proceed as for Lemma 8 in \cite{Lundholm-Solovej:anyon}
		(cp. also \cite{Dyson-Lenard:67} 
		and the proof of Theorem \ref{thm:anyon_gas} above), 
		inserting the partition of unity
		\begin{equation} \label{partition_of_unity_Q}
			1 = \sum_{A \subseteq \{1,\ldots,N\}} \prod_{l \in A} \chi_Q(x_l) \prod_{l \notin A} \chi_{Q^c}(x_l)
		\end{equation}
		in \eqref{local_kinetic_energy_H}:
		\begin{multline*}
			\!\!\!\!\!\!
			\sum_A \int_{\R^N} \frac{1}{2} \sum_{j \in A} \left( 
				|\partial_j \psi|^2 + \sum_{(j\neq)k=1}^N V_{\uH}(x_j-x_k) |\psi|^2
				\right) \prod_{l \in A} \chi_Q(x_l) \prod_{l \notin A} \chi_{Q^c}(x_l) \,d\sx \\
			\!\ge \sum_A \int_{(Q^c)^{N-|A|}} \int_{Q^{|A|}} \!\frac{1}{2} \!\left(
				\sum_{j \in A} |\partial_j \psi|^2 
				+ \!\!\sum_{j\neq k \in A} V_{\uH}(x_j - x_k) |\psi|^2
				\right) \prod_{l \in A} dx_l \prod_{l \notin A} dx_l \\
			\ge \sum_A (|A|-1) \frac{\xi_{\uH}(\alpha)^2}{|Q|^2} 
				\int_{(Q^c)^{N-|A|}} \int_{Q^{|A|}} |\psi|^2 \prod_{l \in A} dx_l \prod_{l \notin A} dx_l \\
			= \frac{\xi_{\uH}(\alpha)^2}{|Q|^2} \int_{\R^N} \left( \sum_{j=1}^N \chi_Q(x_j) - 1 \right) |\psi|^2 d\sx,
		\end{multline*}
		hence \eqref{local_exclusion_density_H},
		where we again used $V_{\uH} \ge 0$ and
		in the last step the partition of unity
		\eqref{partition_of_unity_Q}.
	\end{proof}

	In the Lieb-Liniger case, we observe from 
	Lemma \ref{lem:Schroedinger_properties} 
	with \\ $x := |Q|/\delta(R) \ge 1$ that
	\begin{equation} \label{Schroedinger_estimate}
		\frac{ \xi_{\uS}(\eta \delta(R))^2 }{ \delta(R)^2 } 
		\ge \frac{ \xi_{\uS}(\eta |Q|/x)^2 x }{ |Q|^2 } 
		\ge \frac{ \xi_{\uS}(\eta |Q|)^2 }{ |Q|^2 }
		\quad \text{for all $R \in Q$},
	\end{equation}
	and can thereafter proceed exactly as above 
	for the Calogero-Sutherland case.

	\begin{lem}[Local exclusion for LL] \label{lem:local_exclusion_S}
		We have for $\eta \ge 0$
		\begin{equation} \label{local_exclusion_S}
			\int_{Q^n} T_{\uS}^n(\psi;\sx) \,d\sx
			\ \ge \ (n-1) \frac{\xi_{\uS}(\eta |Q|)^2}{|Q|^2} 
				\int_{Q^n} |\psi|^2 \,d\sx,
		\end{equation}
		and hence, defining the 
		local kinetic energy on $Q$ for a normalized 
		$N$-particle wave function
		$\psi \in \Dc_{\uS}^\eta$
		\begin{equation} \label{local_kinetic_energy_S}
			T_{\uS}^Q := \sum_{j=1}^N \int_{\R^N} \frac{1}{2} \left( |\partial_j \psi|^2 
				+ \sum_{(j\neq)k=1}^N V_{\uS}(x_j - x_k) |\psi|^2 \right) \chi_Q(x_j) \,d\sx,
		\end{equation}
		we obtain the local energy bound
		\begin{equation} \label{local_exclusion_density_S}
			T_{\uS}^Q \ \ge \ \frac{\xi_{\uS}(\eta |Q|)^2}{|Q|^2} 
					\left( \int_Q \rho(x) dx \ - 1 \right)_+.
		\end{equation}
	\end{lem}
	\begin{proof}
		Use the estimate \eqref{Schroedinger_estimate}
		in \eqref{relative_kinetic_bound}
		and continue as in the proof of Lemma \ref{lem:local_exclusion_H}.
	\end{proof}

\section{Local uncertainty} \label{sec:local_uncertainty}

	We shall use the following local form of the uncertainty principle
	on a $d$-dimensional cube $Q$.
	
	\begin{lem}[Local uncertainty principle] 
		\label{lem:local_uncertainty}
		For any $\ep \in (0,1)$ we have
		\begin{equation} \label{local_uncertainty}
			T^Q \ \ge \ 
			\frac{C_d'}{2} \ep^{1+4/d} \ \frac{\int_Q \rho^{1 + 2/d}}{(\int_Q \rho)^{2/d}}
				- \frac{C_d'}{2} \left(1 + \left(\frac{\ep}{1-\ep}\right)^{1+4/d} \right)
					\frac{\int_Q \rho}{|Q|^{2/d}},
		\end{equation}
		where the local kinetic energy $T^Q$ 
		denotes either $T_{\uS}^Q$ with $\eta \ge 0$,
		or $T_{\uH}^Q$ with $\alpha \ge 1$,
		or $T_{\uA}^Q$ with any $\alpha \in \R$,
		or simply $T_0^Q$, the free kinetic energy on a $d$-cube $Q$ 
		in any dimension $d \ge 1$
		(i.e. half of the l.h.s. of \eqref{uncertainty-cube}).
	\end{lem}

	This follows from the following lemma
	concerning the free kinetic energy for an arbitrary 
	$N$-particle wave function,
	which was proved 
	as Theorem 14 in \cite{Lundholm-Solovej:anyon}
	using methods from \cite{Rumin:10}.

	\begin{lem}[Uncertainty on a cube] \label{lem:uncertainty-cube}
		Let $Q$ be a cube in $\R^d$ with volume $|Q|$,
		and let $u \in H^1(\R^{dN})$ be an arbitrary $N$-particle wave function.
		Then 
		\begin{equation} \label{uncertainty-cube}
			\sum_{j=1}^N \int_{\R^{dN}} |\nabla_j u|^2 \,\chi_Q(\bx_j) \,d\sx
			\ \ge \ \frac{C_d'}{(\int_Q \rho)^{2/d}} \int_Q \left[
				\rho(\bx)^{\frac{1}{2}} - \left( \frac{\int_Q \rho}{|Q|} \right)^{\frac{1}{2}}
				\right]_+^{\frac{2(d+2)}{d}} \!\! d\bx,
		\end{equation}
		where $\rho(\bx) := \sum_{j=1}^N \int_{\R^{d(N-1)}} 
			|u(\bx_1,\ldots,\bx_{j-1},\bx,\bx_{j+1},\ldots,\bx_N)|^2 \prod_{k \neq j} d\bx_k$.
	\end{lem}
	
	The constant $C_d' := d^2 C_d^{-\frac{2}{d}} \big/ (d+2)(d+4)$,
	with $C_d$ given below,
	enters via a bound for the growth of the sequence of eigenvalues 
	$\{ \lambda_k \}_{k=0}^\infty$ 
	for the Neumann Laplacian on $Q$:
	\begin{equation} \label{Neumann_eigenvalue_bound}
		\sum_{0 < \lambda_k < e} \frac{2^d}{|Q|} \le C_d \, e^{\frac{d}{2}},
		\qquad e \ge 0.
	\end{equation}
	In the case $d=1$ we have $\lambda_k = \frac{\pi^2}{|Q|^2} k^2$,
	$k=0,1,2,\ldots$, and
	$$
		\#\{ k : 0 < \lambda_k < e \}
		= \#\{ k \in \Z : 0 < k < e^{\frac{1}{2}} |Q|/\pi \}
		\ \le \ e^{\frac{1}{2}} |Q|/\pi,
	$$
	and hence we take $C_1 := 2/\pi$ and $C_1' = \pi^2/60$.
	In the case $d \ge 2$ we have 
	$\lambda_{\bk} = \pi^2 |\bk|^2 / |Q|^{\frac{2}{d}}$,
	$\bk \in \Z_{\ge 0}^d$, and by geometric considerations
	$$
		\#\{ \bk : 0 < \lambda_{\bk} < e \}
		= \#\{ \bk \in \Z_{\ge 0}^d : 0 < |\bk| < e^{\frac{1}{2}} |Q|^{\frac{1}{d}}/\pi \}
		\ \le \ d (e^{\frac{1}{2}} |Q|^{\frac{1}{d}}/\pi)^d,
	$$
	i.e. the optimal constant 
	in \eqref{Neumann_eigenvalue_bound} is
	$C_d := d2^d/\pi^d$, so 
	\begin{equation} \label{optimal_Cprime}
		C_d' = \frac{\pi^2}{4} \frac{d^{2-2/d}}{(d+2)(d+4)}.
	\end{equation}
	In particular, $C_2 = 8/\pi^2$ and $C_2' = \pi^2/48$.
	We remark that these constants are probably far from the sharp ones 
	in the inequality \eqref{uncertainty-cube}.

	\begin{proof}[Proof of Lemma \ref{lem:local_uncertainty}]
		We reduce to the case of free bosons,
		$T^Q \ge T_0^Q$, by discarding the statistics potentials
		or using the diamagnetic inequality (Lemma \ref{lem:diamagnetic_inequality}) 
		$|D_j \psi| \ge |\nabla_j |\psi||$ for anyons 
		(which also preserves $\rho$).
		Then,
		from Lemma \ref{lem:uncertainty-cube}, 
		$$
			T_0^Q
			\ \ge \ \frac{1}{2} \frac{C_d'}{(\int_Q \rho)^{2/d}} \int_Q \left[
				\rho(\bx)^{\frac{1}{2}} - \left( \frac{\int_Q \rho}{|Q|} \right)^{\frac{1}{2}}
				\right]_+^{2+4/d} d\bx,
		$$
		with 
		\begin{multline*}
			\int_Q \left[
				\rho(\bx)^{\frac{1}{2}} - \left( \frac{\int_Q \rho}{|Q|} \right)^{\frac{1}{2}}
				\right]_+^{2+4/d} d\bx \\
			\ge \int_Q \left|
				\rho(\bx)^{\frac{1}{2}} - \left( \frac{\int_Q \rho}{|Q|} \right)^{\frac{1}{2}}
				\right|^{2+4/d} d\bx \ 
				- \int_Q \left( \frac{\int_Q \rho}{|Q|} \right)^{1+2/d} d\bx \\
			= \left\| \rho^{\frac{1}{2}} - \left( \frac{\int_Q \rho}{|Q|} \right)^{\frac{1}{2}} \right\|^{2+4/d}_{2+4/d} 
				- \frac{(\int_Q \rho)^{1+2/d}}{|Q|^{2/d}},
		\end{multline*}
		where the norm $\| \cdot \|_p$ is that of $L^p(Q)$.
		The first term is bounded below by 
		$\left( \| \rho^{\frac{1}{2}} \|_{2+4/d} 
			- \| (\int_Q \rho / |Q|)^{\frac{1}{2}} \|_{2+4/d} \right)^{2+4/d}$
		using the triangle inequality. 
		Furthermore, by convexity we have for any $a,b \in \R$,
		$\ep \in (0,1)$, and $p \ge 1$ that
		$$
			\left( \ep a + (1-\ep) b \right)^p \le \ep a^p + (1-\ep) b^p,
		$$
		and hence with $a = A-B$ and $b = \frac{\ep}{1-\ep}B$,
		$$
			(A-B)^p \ge \ep^{p-1} A^p - \left( \frac{\ep}{1-\ep} \right)^{p-1} B^p.
		$$
		Applying this inequality to the norms
		above with $p=2+4/d$, we finally arrive at 
		\eqref{local_uncertainty}.
	\end{proof}

	We remark that bounds of the form \eqref{local_uncertainty}
	could also be obtained by application of standard Sobolev and Poincar\'e
	inequalities, but would then need to be dealt with differently for
	$d = 1,2$ and $d \ge 3$ 
	(see \cite{Frank-Seiringer:12} concerning the latter case).
	For $d=2$ an alternative form of the local uncertainty principle
	than \eqref{local_uncertainty}
	was given and used in \cite{Lundholm-Solovej:anyon}.

\section{Lieb-Thirring inequalities} \label{sec:Lieb-Thirring}

\subsection{A Lieb-Thirring inequality for anyons}

	In \cite{Lundholm-Solovej:anyon} we introduced an approach to
	Lieb-Thirring inequalities based on local exclusion and uncertainty,
	resulting in the following energy bounds for anyons in $\R^2$.

	\begin{thm}[Kinetic energy inequality for anyons] \label{thm:kinetic-LT-anyons}
		Let $\psi \in \Dc_{\uA}^\alpha$, $\alpha \in \R$,
		be a normalized $N$-anyon wave function 
		on $\R^2$. 
		Then
		\begin{equation} \label{kinetic-LT-anyons}
			\int_{\R^{2N}} T_{\uA}^N(\psi;\sx) \,d\sx
			\ \ge \ C_{\uA} C_{\alpha,N}^2 \int_{\R^2} \rho(\bx)^2 \,d\bx,
		\end{equation}
		for some positive constant $10^{-4} \le C_{\uA} \le \pi$.
	\end{thm}
	\begin{cor}[Lieb-Thirring inequality for anyons]
		Let $\psi \in \Dc_{\uA}^\alpha$  be a normalized $N$-anyon wave function 
		and $V$ a real-valued potential on $\R^2$. Then
		\begin{equation} \label{LT-anyons}
			\int_{\R^{2N}} \left( T_{\uA}^N(\psi;\sx) + \sum_{j=1}^N V(\bx_j)|\psi|^2 \right) \,d\sx
			\ \ge \ -C_{\uA}' C_{\alpha,N}^{-2} \int_{\R^2} |V_-(\bx)|^2 \,d\bx,
		\end{equation}
		for a positive constant $C_{\uA}' = (4C_{\uA})^{-1}$.
	\end{cor}

	The rough lower bound $C_{\uA} \ge 10^{-4}$ that is given here
	for the optimal constant $C_{\uA}$ in \eqref{kinetic-LT-anyons} can be
	computed in a similar way as we do below for the Lieb-Liniger case
	(in the proof of the corresponding 
	Theorem 11 in \cite{Lundholm-Solovej:anyon}
	one can take e.g. $\epsilon = 1/7$ 
	and $\kappa = C_{\alpha,N}^2/12$).
	Note that from \eqref{kinetic-LT-anyons} follows also 
	the bound for the free anyon gas 
	(see \cite{Lundholm-Solovej:anyon})
	\begin{equation} \label{anyon_gas_energy}
		\frac{T_{\uA}}{|\Omega|} \ge C_{\uA} C_{\alpha,N}^2 \bar{\rho}^2,
	\end{equation}
	with mean density $\bar{\rho} = N/|\Omega|$,
	assuming that $\rho$ is supported on a domain $\Omega \subseteq \R^2$.
	Comparing with \eqref{fermion_gas_energy} 
	it is clear
	that the optimal constant in 
	\eqref{kinetic-LT-anyons} and \eqref{anyon_gas_energy}
	satisfies $C_{\uA} \le C_{\uA}^{\,\textup{cl}} := \pi$,
	which is the exact constant for fermions in the semi-classical approximation.
	We should also compare with Theorem \ref{thm:anyon_gas}
	which produced a significantly better numerical bound for the 
	constant for the anyon gas.

\subsection{Lieb-Thirring inequalities in one dimension}

	By proceeding as for the anyons,
	combining local uncertainty with local exclusion, 
	we can prove the following 
	for intermediate statistics in one dimension,
	summarized in Theorem \ref{thm:main}.
	
	\begin{thm}[Lieb-Thirring inequalities for 1D Lieb-Liniger] \label{thm:LT-Schroedinger}
		For $\psi \in \Dc_{\uS}^\eta$ with $\eta \ge 0$ and $N \ge 1$
		we have
		\begin{equation} \label{kinetic-LT-Schroedinger}
			\int_{\R^N} T_{\uS}^N(\psi;\sx) \,d\sx
			\ \ge \ C_{\uS} \int_{\R} \xi_{\uS}(2\eta/\rho^*(x))^2 \rho(x)^3 \,dx,
		\end{equation}
		for some positive constant 
		$3 \cdot 10^{-5} \le C_{\uS} \le 2/3$,
		where $\rho^*$ is
		the Hardy-Littlewood maximal function of $\rho$
		(see \eqref{Hardy-Littlewood-rho}).
		In particular, if $\frac{\int_Q \rho}{|Q|} \le \gamma \bar{\rho}$
		for all intervals $Q$ and some $\gamma > 0$, then
		\begin{equation} \label{kinetic-LT-SchroedingerBd}
			\int_{\R^N} T_{\uS}^N(\psi;\sx) \,d\sx
			\ \ge \ C_{\uS} \,\xi_{\uS}(2\eta/(\gamma\bar{\rho}))^2 \int_{\R} \rho(x)^3 \,dx,
		\end{equation}
		\begin{equation} \label{LT-SchroedingerBd}
			\int_{\R^N} \left( T_{\uS}^N(\psi;\sx) + \sum_{j=1}^N V(x_j) |\psi|^2 \right) d\sx
			\ \ge \ -\frac{C_{\uS}'}{ \xi_{\uS}(2\eta/(\gamma\bar{\rho})) } \int_{\R} |V_-(x)|^{\frac{3}{2}} \,dx,
		\end{equation}
		with $C_{\uS}' := \frac{2}{3}(3C_{\uS})^{-\frac{1}{2}}$, 
		and if $\rho$ is supported on an interval of length $L$
		\begin{equation} \label{energy-SchroedingerBd-gas}
			T_{\uS}/L \ \ge \ C_{\uS} \,\xi_{\uS}(2\eta/(\gamma\bar{\rho}))^2 \bar{\rho}^3,
			\qquad \bar{\rho} := N/L.
		\end{equation}
	\end{thm}

	\begin{proof}
		For $N=1$ and $N=2$ we use $T_{\uS} \ge T_0$
		and Lemma \ref{lem:uncertainty-cube}
		(which generalizes to $Q=\R$) with a resulting constant
		$C_1'/8 \ge 0.020$.
		
		For $N \ge 3$ we can consider the kinetic energy $T_{\uS}^Q$
		on an arbitrary finite interval $Q_0 \subset \R$ 
		s.t. $\int_{Q_0} \rho \ge 2$.
		We split the interval $Q_0$ 
		in halves iteratively, organizing the resulting subintervals $Q$
		in a (full binary) tree $\mathbb{T}$; cp. \cite{Lundholm-Solovej:anyon}.
		The procedure can be arranged so that $Q_0$ is finally
		covered by intervals $Q_B$ marked B 
		s.t. $2 \le \int_{Q_B} \rho < 4$, and
		$Q_A$ marked A 
		s.t. $0 \le \int_{Q_A} \rho < 2$,
		sitting at the leaves of the tree
		and s.t. at least one B-interval is at the highest level
		of every branch of the tree.
		On the B-intervals we use local exclusion, 
		\eqref{local_exclusion_density_S}, 
		together with local uncertainty,
		Lemma \ref{lem:local_uncertainty},
		to obtain
		\begin{equation}
			T_{\uS}^{Q_B} \ge 
			\kappa \frac{C_1'}{2} \ep^5 \frac{\int_{Q_B} \rho^3}{4^2}
			- \kappa \frac{C_1'}{2} \left( 1 + \left(\frac{\ep}{1-\ep}\right)^5 \right) \frac{4}{|Q_B|^2}
			+ (1-\kappa) \frac{\xi_{\uS}(\eta |Q_B|)^2}{|Q_B|^2},
		\end{equation}
		for any $\ep,\kappa \in (0,1)$.
		For simplicity we set $\ep = \frac{1}{2}$ and 
		$\kappa = \frac{1}{2} \xi_{\uS}(\eta|Q_B|)^2/(\pi/2)^2 \\ \le \frac{1}{2}$
		to find
		\begin{equation} \label{B-bound}
			T_{\uS}^{Q_B} \ge \xi_{\uS}(\eta |Q_B|)^2 
			\left( c_1 \int_{Q_B} \rho^3 + \frac{c_2}{|Q_B|^2} \right),
		\end{equation}
		with $c_1 = 2/\pi^2 \cdot C_1'/2^{10} = 2^{-11}/15$
		and $c_2 = -2/\pi^2 \cdot 4C_1' + 1/2 = 11/30$.
		Note by monotonicity of $\xi_{\uS}$ that
		$$
			\xi_{\uS}(\eta |Q_B|) \ge \xi_{\uS}(2\eta / \tilde{\rho}(x)),
			\quad \text{where} \quad
			\tilde{\rho}|_{Q_B} := \frac{\int_{Q_B}\rho}{|Q_B|},
		$$
		i.e. $\tilde{\rho}$ is defined to be the mean of $\rho$
		on each B-interval.
		The A-intervals are further divided into a subclass A$_2$ on which
		$$
			\int_{Q_A} \rho^3 > c \frac{(\int_{Q_A} \rho)^3}{|Q_A|^2}
		$$
		with $c := 2^7$,
		so that by Lemma \ref{lem:local_uncertainty} 
		(again using $\ep = \frac{1}{2}$)
		\begin{equation} \label{A2-bound}
			T_{Q_{A_2}} \ge c_3 \int_{Q_{A_2}} \rho^3 
			= \frac{4c_3}{\pi^2} \int_{Q_{A_2}} \xi_{\uS}(2\eta / \tilde{\rho})^2 \rho^3 ,
			\quad \text{where} \quad
			\tilde{\rho}|_{Q_{A_2}} := 0,
		\end{equation}
		$c_3 = \frac{1}{2} C_1'/2^8 = \pi^2 c_1$,
		and a remaining subclass A$_1$ for which
		$$
			\int_{Q_A} \rho^3 \le c\frac{(\int_{Q_A} \rho)^3}{|Q_A|^2}.
		$$
		Consider the set $\mathcal{A}_1(Q_B)$
		of such intervals $Q_{A_1}$ which can be found by going
		back in the tree $\mathbb{T}$ from a fixed B-interval $Q_B$
		at level $k \in \N$ (possibly all the way to $Q_0$), 
		and then one step forward.
		On each level $1 \le j \le k$ there is at most one such 
		interval which we denote by $Q_j$
		(and otherwise we can define $|Q_j| := |Q_0|/2^j$ 
		for the below expressions to make sense).
		Using Lemma \ref{lem:Schroedinger_properties} we have
		$$
			\xi_{\uS}\left( \eta|Q_B|/(|Q_B|/|Q_j|) \right)^2 \frac{|Q_B|}{|Q_j|} 
			\le \xi_{\uS}(\eta |Q_B|)^2
		$$
		and hence
		\begin{multline*}
			\sum_{j=1}^k \xi_{\uS}(\eta|Q_j|)^2 \frac{|Q_B|^2}{|Q_j|^2}
			\le \sum_{j=1}^k \xi_{\uS}(\eta|Q_B|)^2 \frac{|Q_B|}{|Q_j|}
			= \xi_{\uS}(\eta|Q_B|)^2 \sum_{j=1}^k \frac{2^{-k}}{2^{-j}} \\
			\le 2\xi_{\uS}(\eta|Q_B|)^2.
		\end{multline*}
		Defining
		$\tilde{\rho}|_{Q_{A_1}} := \int_{Q_{A_1}} \rho/|Q_{A_1}|$
		to be the mean also on A$_1$-intervals,
		we have by 
		$0 \le \tilde{\rho} \le 2/|Q_{A_1}|$ 
		and monotonicity of $\xi_{\uS}(\eta/x)^2x^3$ that
		$$
			\int_{Q_{A_1}} \xi_{\uS}(2\eta/\tilde{\rho})^2 \rho^3
			\le \xi_{\uS}(2\eta/\tilde{\rho})^2 \,c \tilde{\rho}^3 |Q_{A_1}|
			\le c \,\xi_{\uS}(\eta|Q_{A_1}|)^2 \frac{2^3}{|Q_{A_1}|^2},
		$$
		and therefore by the above
		$$
			\sum_{Q_{A_1} \in \mathcal{A}_1(Q_B)} \int_{Q_{A_1}} \xi_{\uS}(2\eta/\tilde{\rho})^2 \rho^3
			\le \sum_{j=1}^k 8c \frac{\xi_{\uS}(\eta|Q_j|)^2}{|Q_j|^2}
			\le 16c \frac{\xi_{\uS}(\eta|Q_B|)^2}{|Q_B|^2},
		$$
		for every B-interval $Q_B$.
		In other words the energy on all intervals with almost constant
		low density is dominated by that from exclusion on the B-intervals.
		Using \eqref{B-bound} we then have
		\begin{equation} \label{B-A1-bound}
			T_{\uS}^{Q_B} \ge c_1 \int_{Q_B} \xi_{\uS}(2\eta/\tilde{\rho})^2 \rho^3
			+ c_2' \sum_{Q_{A_1} \in \mathcal{A}_1(Q_B)} 
				\int_{Q_{A_1}} \xi_{\uS}(2\eta/\tilde{\rho})^2 \rho^3,
		\end{equation}
		$c_2' = c_2/(16c) = 2^{-12} \cdot 11/15$,
		and hence by \eqref{A2-bound} and \eqref{B-A1-bound}
		\begin{equation} \label{kinetic-LT-Schroedinger-Q_0}
			T_{\uS} \ge T_{\uS}^{Q_0}
			= \sum_{Q_{A,B} \in \mathbb{T}} T_{\uS}^{Q_{A,B}} 
			\ge C_{\uS} \int_{Q_0} \xi_{\uS}(2\eta/\tilde{\rho})^2 \rho^3,
		\end{equation}
		where $C_{\uS} := \min\{c_1, c_2', 4c_3/\pi^2\} = c_1 > 3 \cdot 10^{-5}$.
		We now use
		$\tilde{\rho} \le \rho^*$, with 
		\begin{equation} \label{Hardy-Littlewood-rho}
			\rho^*(x) := \sup \left\{ \frac{\int_Q \rho}{|Q|} : \text{$Q$ is a finite interval containing $x$} \right\}
		\end{equation}
		the (uncentered) Hardy-Littlewood maximal function of $\rho$,
		and finally let $Q_0$ approach $\R$ to obtain 
		\eqref{kinetic-LT-Schroedinger}.
		
		The bounds
		\eqref{LT-SchroedingerBd} and \eqref{energy-SchroedingerBd-gas}
		follow in a standard way from 
		\eqref{kinetic-LT-SchroedingerBd}
		(see e.g. Theorem 4.3 in \cite{Lieb-Seiringer:10}).
		The bound $C_{\uS} \le 2/3$ for the optimal constant in
		\eqref{kinetic-LT-Schroedinger} follows by taking the limit
		$\eta \to +\infty$ and comparing \eqref{energy-SchroedingerBd-gas}
		with the semiclassics for fermions.
	\end{proof}

	\begin{thm}[Lieb-Thirring inequalities for 1D Calogero-Sutherland] \label{thm:LT-Heisenberg}
		Let $\psi \in \Dc_{\uH}^\alpha$ 
		with $\alpha \ge 1$ and $N \ge 2$.
		Given any finite interval $Q_0$ s.t. $\int_{Q_0} \rho \ge 2$,
		we have
		\begin{equation} \label{kinetic-LT-Heisenberg-interval}
			T_{\uH}^{Q_0}
			\ \ge \ C_{\uH} \, \xi_{\uH}(\alpha)^2 \int_{Q_0} \tilde{\rho}(x)^3 \,dx
			\ \ge \ C_{\uH} \, \xi_{\uH}(\alpha)^2 \frac{(\int_{Q_0} \rho(x) \,dx)^3}{|Q_0|^2},
		\end{equation}
		for some positive constant 
		$1/32 \le C_{\uH} \le 2/3$,
		and $\tilde{\rho}$ defined below (see \eqref{rho-tilde})
		as a piecewise constant approximation to $\rho$. 
		In particular,
		\begin{equation} \label{kinetic-LT-Heisenberg}
			T_{\uH} = \int_{\R^N} T_{\uH}^N(\psi;\sx) \,d\sx
			\ \ge \ C_{\uH} \, \xi_{\uH}(\alpha)^2 \int_{\R} \tilde{\rho}(x)^3 \,dx,
		\end{equation}
		and if $\rho$ is confined to a length $L$
		\begin{equation} \label{energy-Heisenberg-gas}
			T_{\uH}/L \ \ge \ C_{\uH} \,\xi_{\uH}(\alpha)^2 \bar{\rho}^3,
			\qquad \bar{\rho} := N/L.
		\end{equation}
	\end{thm}
	
	Here the assumption $N>1$ is important because of the
	dependence on $\alpha$ in the bound \eqref{kinetic-LT-Heisenberg}.
	Note that the integral in the r.h.s. of \eqref{kinetic-LT-Heisenberg} 
	actually becomes a Riemann sum approximation of 
	$\int_{\R} \rho(x)^3 \,dx$ 
	in the limit of many particles with relatively uniform density.
	Also note that the usual Lieb-Thirring inequality 
	\eqref{kinetic-energy-inequality} for fermions
	(with a constant independent of $\alpha$)
	is valid as a lower bound for $T_{\uH}$ for all $N$ and $\alpha \ge 1$.
	However, we will prove below that $\tilde{\rho}$ cannot be replaced
	with the actual density $\rho$ 
	in the r.h.s. of \eqref{kinetic-LT-Heisenberg}
	as long as $C_{\uH}$ is a positive constant independent of $\alpha$,
	using that $\xi_{\uH}(\alpha)$ can become arbitrarily large 
	with $\alpha \to \infty$.

	\begin{proof}
		We split an arbitrary finite interval 
		$Q_0$ s.t. $\int_{Q_0}\rho \ge 2$
		iteratively as in the proof of the previous theorem
		(in the case $N=2$ we can approximate this condition
		arbitrarily well with a finite interval, and some
		minor adjustments to the constants below 
		will need to be incorporated accordingly).
		In the Calogero-Sutherland case we cannot combine 
		\eqref{local_exclusion_density_H} 
		with Lemma \ref{lem:local_uncertainty} 
		to yield a uniform bound w.r.t. $\alpha$ 
		(better than the one for fermions)
		since $\xi_{\uH}(\alpha) \to \infty$ as $\alpha \to \infty$.
		Instead we need to rely solely on energy from exclusion
		(also resulting in a slightly better lower bound 
		for the constant $C_{\uH}$).
		We define $\tilde{\rho}|_{Q_0^c} := 0$ and
		\begin{equation} \label{rho-tilde}
			\tilde{\rho}|_Q := \frac{\int_Q\rho}{|Q|}
		\end{equation}
		for all A- and B-intervals $Q \in \mathbb{T}$, 
		and use for $0<\kappa<1$
		\begin{equation}
			T_{\uH}^{Q_B} \ge \frac{\xi_{\uH}(\alpha)^2}{|Q_B|^2} 
			\left( \int_{Q_B} \rho  \ - 1 \right)_+
			\ge \xi_{\uH}(\alpha)^2 \left( 
				\kappa \frac{3}{4^3} \int_{Q_B} \tilde{\rho}^3
				+ (1-\kappa) \frac{1}{|Q_B|^2}
				\right)
		\end{equation}
		(with $x-1 \ge 3x^3/4^3$ when $2 \le x := \int_{Q_B} \rho \le 4$)
		for the B-intervals, and
		$$
			\sum_{Q_A \in \mathcal{A}(Q_B)} \int_{Q_A} \tilde{\rho}^3 
			= \sum_{j=1}^k \frac{(\int_{Q_j} \rho)^3}{|Q_j|^2}
			\le \sum_{j=1}^k \frac{8}{4^{-j}|Q_0|^2}
			\le \frac{8}{3} \frac{4^{k+1}}{|Q_0|^2}
			= \frac{2^5/3}{|Q_B|^2}
		$$
		for all the A-intervals associated to a B-interval $Q_B$ 
		at level $k$ in $\mathbb{T}$ 
		(similarly as for $\mathcal{A}_1(Q_B)$ in the previous proof).
		Hence,
		$$
			T_{\uH}^{Q_B} \ge \xi_{\uH}(\alpha)^2 \left( 
			\kappa \frac{3}{2^6} \int_{Q_B} \tilde{\rho}^3
			+ (1-\kappa) \frac{3}{2^5} \sum_{Q_A \in \mathcal{A}(Q_B)} 
				\int_{Q_A} \tilde{\rho}^3 
				\right),
		$$
		and
		$$
			T_{\uH} \ge T_{\uH}^{Q_0}
			\ge \sum_{Q_B \in \mathbb{T}} T_{\uH}^{Q_B} 
			\ge C_{\uH} \,\xi_{\uH}(\alpha)^2 \int_{Q_0} \tilde{\rho}^3,
		$$
		where $\kappa := 2/3$, $C_{\uH} := 2^{-5}$.
		
		For \eqref{kinetic-LT-Heisenberg-interval} 
		and \eqref{energy-Heisenberg-gas} we use that
		$\int_{Q_0} \tilde{\rho}^3 \ge |Q_0|^{-2} (\int_{Q_0} \tilde{\rho} )^3$
		and $\int_{Q_0} \tilde{\rho} = \int_{Q_0} \rho$.
		The upper bound on the optimal constant $C_{\uH}$ follows by comparing
		\eqref{energy-Heisenberg-gas} with $\alpha=1$ to the semiclassical
		limit for fermions.
	\end{proof}
	
	\begin{thm}[Necessity of $\tilde{\rho}$ in \eqref{kinetic-LT-Heisenberg}]
		Assume that for some nonnegative constant $C_{\uH}$ 
		the inequality
		\begin{equation} \label{kinetic-LT-Heisenberg-rho}
			\int_{\R^N} T_{\uH}^N(\psi;\sx) \,d\sx
			\ \ge \ C_{\uH} \, \xi_{\uH}(\alpha)^2 \int_{\R} \rho(x)^3 \,dx
		\end{equation}
		holds for all $\psi \in \Dc_{\uH}^\alpha$ and $\alpha \ge 1$.
		Then we must have $C_{\uH} = 0$.
	\end{thm}
	\begin{proof}
		The idea is to spread out the particles 
		one by one but with high individual localization,
		and then compare the behavior of the 
		left and right hand sides of \eqref{kinetic-LT-Heisenberg-rho}
		as $\alpha \to \infty$.
		Let $\varphi \in C_c^\infty([-1,1];\R_{\ge 0})$
		be a bump function normalized s.t. $\int_{-1}^1 \varphi^2 \,dx = 1$,
		and define
		$\varphi_\ep(x) := \ep^{-\frac{1}{2}} \varphi(x/\ep)$,
		and the completely symmetric $N$-particle wave function
		$$
			\psi(\sx) := \frac{1}{\sqrt{N!}} \sum_{\sigma \in S_N} \prod_{j=1}^N \varphi_\ep(x_j - \sigma(j)).
		$$
		Assuming $0<\ep < 1/3$ we have on the support of $\psi$ that
		$|x_j - x_k| > 1/3$ for all $j \neq k$, so this wave function 
		is certainly in the domain of $T_{\uH}$.
		It is furthermore normalized:
		$$
			\int_{\R^N} |\psi|^2 \,d\sx 
			= \frac{1}{N!} \sum_{\sigma,\tau \in S_N} \int_{\R^N} \prod_{j=1}^N 
				\underbrace{ \varphi_\ep(x_j - \sigma(j)) \varphi_\ep(x_j - \tau(j)) }_{ \delta_{\sigma(j),\tau(j)} \,\varphi_\ep(x_j - \sigma(j))^2 } 
				\,d\sx
			= 1,
		$$
		with one-particle density
		\begin{multline*}
			\rho(x) = N \int_{\R^{N-1}} |\psi(x,x_2,\ldots,x_N)|^2 \,d\sx' \\
			= \frac{N}{N!} \sum_{\sigma,\tau \in S_N} \int_{\R^{N-1}} 
				\varphi_\ep(x - \sigma(1)) \varphi_\ep(x - \tau(1))
				\prod_{j=2}^N \varphi_\ep(x_j - \sigma(j)) \varphi_\ep(x_j - \tau(j))
				\,d\sx' \\
			= \sum_{j=1}^N \varphi_\ep(x - j)^2,
		\end{multline*}
		supported on the interval $[0,N+1]$,
		and s.t. $\rho(x)^3 = \sum_{j=1}^N \varphi_\ep(x-j)^6$.
		Furthermore,
		$$
			\int_{\R^N} \sum_{j<k} \frac{\alpha(\alpha-1)}{(x_j - x_k)^2} \,|\psi|^2 \,d\sx
			< 9 \alpha(\alpha-1) \binom{N}{2},
		$$
		and
		$$
			\partial_k \psi = \frac{\ep^{-1}}{\sqrt{N!}} \sum_{\sigma \in S_N} \ep^{-\frac{1}{2}} \varphi'((x_j - \sigma(k))/\ep) \prod_{j\neq k} \varphi_\ep(x_j - \sigma(j)),
		$$
		implying
		$
			\int_{\R^N} |\partial_k \psi|^2 \,d\sx = \ep^{-2} \int_{-1}^1 (\varphi')^2 \,dx.
		$
		Hence,
		\begin{equation} \label{TH_lhs}
			T_{\uH} < \ep^{-2} N \int_{-1}^1 (\varphi')^2 \,dx + \frac{9}{2} \alpha(\alpha-1) N(N-1),
		\end{equation}
		while the r.h.s. of \eqref{kinetic-LT-Heisenberg-rho} is
		\begin{equation} \label{TH_rhs}
			C_{\uH} \,\xi_{\uH}(\alpha)^2 \sum_{j=1}^N \int_{\R} \ep^{-3} \varphi((x - j)/\ep)^6 \,dx
			= \ep^{-2} N C_{\uH} \,\xi_{\uH}(\alpha)^2 \int_{-1}^1 \varphi^6 \,dx.
		\end{equation}
		Now, taking $\ep$ depending on $N$ so small that 
		$\ep^{-2} N C_{\uH} \int_{-1}^1 \varphi^6 \,dx > 9N(N-1)$,
		we find using \eqref{approx_H_large},
		$\xi_{\uH}(\alpha)^2 > \alpha(\alpha-1)$,
		that \eqref{TH_rhs} is strictly greater than \eqref{TH_lhs}
		for $\alpha$ sufficiently large, 
		unless $C_{\uH} = 0$.
		
		Note on the other hand that $\tilde{\rho}$, 
		and hence the r.h.s. of \eqref{kinetic-LT-Heisenberg},
		is independent of $\ep$ for the above choice of $\psi$.
	\end{proof}

\section{Some applications} \label{sec:applications}

\subsection{Many anyons in a harmonic oscillator potential} \label{sec:anyon_osc}

	As an application of the kinetic energy inequality
	\eqref{kinetic-LT-anyons} for anyons, we can consider the ground 
	state energy of $N$ anyons in a harmonic oscillator potential,
	a problem which has been discussed substantially in the literature.
	The Hamiltonian operator is in our conventions given by
	$$
		H = \sum_{j=1}^N \left( \frac{1}{2} D_j \cdot D_j + V(\bx_j) \right),
		\qquad
		V(\bx) := \frac{\omega^2}{2} |\bx|^2.
	$$ 
	For $N=2$ the exact spectrum was computed already in 
	\cite{Leinaas-Myrheim:77}, while for $N \ge 3$ the ground
	state energy is only known exactly for fermions $\alpha=1$ and in
	a small neighborhood around bosons $\alpha=0$. 
	The lower part of the spectrum has been computed numerically 
	for the full range of $\alpha \in [0,1]$ for $N=3,4$ 
	\cite{Sporre-Verbaarschot-Zahed:91,Murthy_et_al:91,Sporre-Verbaarschot-Zahed:92}
	and for general $N$ there are some known families of
	exact but excited or singular eigenstates 
	\cite{Wu:84,Chou:91,Bhaduri_et_al:92}.
	Some general features extend to arbitrary $N$,
	such as large numbers of level crossings in the ground state
	as $N$ is fixed large and $\alpha$ varied \cite{Chitra-Sen:92}.
	See \cite{Khare:05} for a more recent summary 
	of the status of this problem.

	With the kinetic energy inequality \eqref{kinetic-LT-anyons} at hand, 
	we would like to find a lower bound for the quadratic
	form
	\begin{multline} \label{harm_osc_qf}
		\langle \psi, H\psi \rangle 
		= \int_{\R^{2N}} \sum_j \left( \frac{1}{2} |D_j \psi|^2 + V(\bx_j) |\psi|^2 \right) \,d\sx \\
		\ge \int_{\R^2} \left( C_{\uA} C_{\alpha,N}^2 \rho(\bx)^2 + \frac{\omega^2}{2} |\bx|^2 \rho(\bx) \right) d\bx
	\end{multline}
	subject to the conditions $\int_{\R^2} \rho = N$
	and $\rho \ge 0$.
	We extremize the functional
	$$
		F[\rho,\lambda] := \int_{\R^d} \left( C_{\uA} C_{\alpha,N}^2 \rho(\bx)^2 
			+ \frac{\omega^2}{2} |\bx|^2 \rho(\bx) - \lambda \rho(\bx) \right) d\bx 
			+ \lambda N
	$$
	to find
	$$
		2C_{\uA} C_{\alpha,N}^2 \rho(\bx) + \frac{\omega^2}{2} |\bx|^2 - \lambda = 0,
		\qquad \text{and} \qquad
		N = \int_{\R^2} \rho(\bx) \,d\bx,
	$$
	i.e.
	$$
		\rho(\bx) = \frac{[\lambda - \omega^2 |\bx|^2/2]_+}{2C_{\uA} C_{\alpha,N}^2},
	$$
	with 
	$$
		2C_{\uA} C_{\alpha,N}^2 N 
		= 2\pi \int_0^{\sqrt{2\lambda}/\omega} 
			(\lambda - \omega^2 r^2/2) r\,dr
		= 2\pi\left[ \lambda\frac{r^2}{2} - \omega^2\frac{r^4}{8} \right]_0^{\sqrt{2\lambda}/\omega}
		= \frac{\pi\lambda^2}{\omega^2},
	$$
	and hence $\lambda = \omega C_{\alpha,N} \sqrt{2C_{\uA} N/\pi}$.
	Plugging in this form for $\rho$, 
	the r.h.s. of \eqref{harm_osc_qf} is
	\begin{multline*}
		\frac{2\pi}{2C_{\uA} C_{\alpha,N}^2} \int_0^{\sqrt{2\lambda}/\omega} \left(
			\frac{1}{2}(\lambda - \omega^2 r^2/2)^2
			+ \frac{\omega^2}{2} r^2 (\lambda - \omega^2 r^2/2)
			\right) r\,dr \\
		= \frac{\pi}{C_{\uA} C_{\alpha,N}^2} \left[ 
			\frac{\lambda^2}{2} \frac{r^2}{2}
			+ \frac{\omega^4}{4} \left( \frac{1}{2} - 1 \right) \frac{r^6}{6} 
			\right]_0^{\sqrt{2\lambda}/\omega} 
		= \frac{\pi}{2C_{\uA} C_{\alpha,N}^2} \left( 
			\frac{\lambda^3}{\omega^2} - \frac{\lambda^3}{3\omega^2} 
			\right) \\
		= \frac{\pi(\omega C_{\alpha,N} \sqrt{2C_{\uA} N} / \sqrt{\pi})^3}{3C_{\uA} C_{\alpha,N}^2 \omega^2},
	\end{multline*}
	and hence
	\begin{equation} \label{harm_osc_bound}
		\langle \psi,H\psi \rangle \ge \frac{1}{3} \sqrt{\frac{8C_{\uA}}{\pi}} C_{\alpha,N} \omega N^{\frac{3}{2}}.
	\end{equation}
	This resulting bound for the energy in terms of 
	$C_{\alpha,N} \omega N^{\frac{3}{2}}$
	allows us to compare the graph of $\alpha \mapsto C_{\alpha,N}$ 
	(the limiting graph for 
	$C_\alpha := \lim_{N \to \infty} C_{\alpha,N}$
	is sketched in Figure 2 in \cite{Lundholm-Solovej:anyon})
	with the previously found exact and numerical spectra 
	for 2--4 anyons
	and the extensions of certain features in such spectra to more anyons
	(note e.g. the partial reflection symmetry about $\alpha=1/2$;
	cp. \cite{Sen:92}).
	The correct ground state energy in the case $\alpha=1$ is
	to leading order
	$\sim \frac{\sqrt{8}}{3} \omega N^{\frac{3}{2}}$.
	We also note that for bosons, which satisfy the weaker kinetic
	energy inequality $T_0 \ge \frac{C_{\uA}}{N}\int_{\R^2} \rho^2$,
	the above bound reduces to
	$$
		\langle \psi,H\psi \rangle
		\ge \frac{1}{3} \sqrt{\frac{8C_{\uA}}{\pi}} \omega N.
	$$
	The correct ground state energy for bosons is of course
	$\omega N$.

	The bound \eqref{harm_osc_bound} improves,
	for odd numerator rational $\alpha$, 
	the previously known best lower bound for the energy
	of $N$ anyons
	in a harmonic oscillator
	\cite{Sen:92,Chitra-Sen:92}:
	\begin{equation} \label{harm_osc_bound_L}
		\langle \psi,H\psi \rangle
		\ge \omega\left( N + \left|L + \alpha\frac{N(N-1)}{2}\right| \right),
	\end{equation}
	where $L$ is the total angular momentum of the state $\psi$;
	$(-i\sum_j \bx_j \wedge \nabla_j)\psi = L\psi$.
	We see from this inequality that the energy 
	actually grows like $N^2$, unless $L \sim -\alpha \binom{N}{2}$ 
	(equality is possible for certain $N$ and $\alpha$)
	for which \eqref{harm_osc_bound_L} 
	reduces to the bosonic bound for the energy.
	The non-trivial dependence 
	on $\alpha$ and $N$
	of these energy bounds, 
	together with other aspects of the spectra mentioned above,
	raises the question whether the limiting constant 
	in the kinetic energy inequality \eqref{kinetic-LT-anyons}
	actually cannot be improved for even numerator and irrational 
	values of $\alpha$
	(for which $C_{\alpha,N} \to 0$ as $N \to \infty$).
	We refer to \cite{Lundholm-Solovej:exclusion} 
	for further discussion on this possibility.

\subsection{Stability of odd-fractional anyonic matter with Coulomb interactions} \label{anyon_stability}

	In the context of the fractional quantum Hall effect it is
	relevant to consider anyons arising in the form of quasiparticles 
	with fractional charge and statistics 
	\cite{Arovas-Schrieffer-Wilczek:84} 
	embedded in a three-dimensional system,
	possibly with an effective 3D Coulomb repulsion.
	Since these are then arising in a material 
	there could also be a large number
	of oppositely charged particles 
	which impose a Coulomb attraction on the anyons.
	We can ask if such an interacting system is 
	stable in the thermodynamic limit, 
	i.e. if the negative ground state energy grows at most linearly with 
	the total number of particles
	(cp. stability for ordinary fermionic matter \cite{Lieb-Seiringer:10}).
	Consider therefore the model Hamiltonian
	$$
		H(N,\sR) := \frac{1}{2m} \sum_{j=1}^N D_j \cdot D_j + V_C(\sx,\sR),
	$$
	$$
		V_C(\sx,\sR) := \sum_{1 \le j<k \le N} \frac{1}{|\bx_j - \bx_k|}
			- \sum_{j=1}^N \sum_{k=1}^K \frac{Z}{|\bx_j - \bR_k|}
			+ \sum_{1 \le j<k \le K} \frac{Z^2}{|\bR_j - \bR_k|},
	$$
	which can be taken to describe 
	$N$ anyons with mass $m$ and unit charge $-1$ moving in the plane $\R^2$,
	together with $K$ static `nuclei' of charge $Z \ge 1$
	at positions $\sR = (\bR_1,\ldots,\bR_K) \in \R^{2K}$,
	and 
	where all particles are
	interacting through electromagnetic forces in 3D.
	The condition $Z \ge 1$ is only technical and can be relaxed,
	while the overall scale of the unit charge can be adjusted with
	the variable mass $m$.

	By the usual electromagnetic screening in 3D, 
	we have the following important lemma to simplify the problem
	(given e.g. as Theorem 5.4 in \cite{Lieb-Seiringer:10}).
	
	\begin{lem}[Baxter's electrostatic inequality]
		$$
			V_C(\sx,\sR) \ge -(2Z+1)\sum_{j=1}^N \frac{1}{d_R(\bx_j)} 
			+ \frac{Z^2}{8} \sum_{k=1}^K \frac{1}{D_k},
		$$
		where
		$D_k := \frac{1}{2} \min_{j \neq k} |\bR_j - \bR_k|$,
		and
		$d_{\sR}(\bx) := \min_{k=1,\ldots,K} |\bx - \bR_k|$.
	\end{lem}

	\noindent
	Hence, we would like to find a lower bound,
	independent of $\sR$, for the quadratic form
	\begin{multline} \label{stability_qf}
		\langle \psi, H(N,\sR)\psi \rangle
		= \int_{\R^{2N}} \left( \frac{1}{2m} \sum_j |D_j\psi|^2 + V_C(\sx,\sR) |\psi|^2 \right) \,d\sx \\
		\ge \frac{1}{2m} \int_{\R^{2N}} \sum_j \left( 
			|D_j\psi|^2 - 2m\frac{2Z+1}{d_{\sR}(\bx_j)} |\psi|^2 \right) \,d\sx
			+ \frac{Z^2}{8} \sum_{k=1}^K \frac{1}{D_k}.
	\end{multline}
	With the Lieb-Thirring inequality \eqref{LT-anyons} for anyons at hand, 
	we will prove the following theorem.
	
	\begin{thm}[Stability of `anyonic matter' for odd-fractional statistics]
		For a normalized $N$-anyon wave function 
		$\psi \in \Dc_{\uA}^\alpha$, with
		$\alpha = \mu/\nu$ an odd numerator reduced fraction, 
		we have the lower bound
		$$
			\langle \psi,H(N,\sR)\psi \rangle \ge -C \,\nu^2 mZ^2 (K+N),
		$$
		for some positive constant $C$.
	\end{thm}
	
	There are some technical difficulties related to the dimensionality.
	In order to handle the non-square-integrability of the one-particle
	potential in \eqref{stability_qf} close to the nuclei, 
	we split the integral into two regions.
	On a region $B_\bep := \bigcup_{k=1}^K B_{\ep_k}(\bR_k)$,
	$0 < \ep_k \le D_k$,
	close to the nuclei we use
	$|D_j \psi| \ge \left| \nabla_j |\psi| \right|$
	together with the following lemma to deduce
	\begin{multline} \label{bound_close_to_nuclei}
		\int_{\R^{2N}} \left( 
			\frac{1}{2}|D_j \psi|^2 - 2m\frac{2Z+1}{d_{\sR}(\bx_j)} |\psi|^2 
			\right) \chi_{B_\bep}(\bx_j) \,d\sx \\
		\ge -2m(2Z+1) \int_{\R^{2N}} 
			\left( 4m(2Z+1) + \frac{2}{\ep(\bx_j)} \right)|\psi|^2 
			\chi_{B_\bep}(\bx_j) \,d\sx,
	\end{multline}
	where $\ep(\bx) := \ep_k$ on 
	the disk $B_{\ep_k}(\bR_k)$ around the $k$th nucleus.

	\begin{lem}
		For $v \in H^1(B_\ep(0))$ and $\ep,\mu>0$ we have
		$$
			\int_{B_\ep(0)} \left( 
				\mu |\nabla v|^2 - \frac{1}{|\brr|}|v|^2
				\right) d\brr
			\ge -\left( \frac{1}{\mu} + \frac{2}{\ep} \right)
				\int_{B_\ep(0)} |v|^2 \,d\brr.
		$$
	\end{lem}
	\begin{proof}
		We simply adapt the proof of Lemma 2 in \cite{Dyson-Lenard:67}
		to two dimensions.
		It is clear that the ratio
		$$
			\int_{B_\ep(0)} \left( 
				\mu |\nabla v|^2 - \frac{1}{|\brr|}|v|^2
				\right) d\brr 
				\Big/ \int_{B_\ep(0)} |v|^2 \,d\brr
		$$
		is minimized by a radially symmetric function $v(r)$,
		and that we may consider a lower bound 
		when the domain is replaced by the annulus
		$B_\ep(0) \setminus B_\delta(0)$ 
		if the bound is uniform in $\delta>0$.
		We have therefore reduced to the Neumann problem
		$$
			-\mu\left( v'' + \frac{v'}{r} \right) - \frac{v}{r} = \lambda v,
			\qquad
			v'(\delta) = v'(\ep) = 0,
		$$
		on $[\delta,\ep]$,
		and would like to show that the lowest eigenvalue
		$\lambda \ge -\frac{1}{\mu} - \frac{2}{\ep}$.
		Noting that the ground state $v$ is nodeless
		and defining $\omega := -v'/v$, we have
		\begin{equation} \label{omega-Neumann}
			\frac{\mu}{r}\partial_r (r\omega) - \mu \omega^2 - \frac{1}{r} = \lambda,
			\qquad
			\omega(\delta) = \omega(\ep) = 0,
		\end{equation}
		and after multiplying by $r$ and integrating over 
		$[\delta,\ep]$,
		\begin{equation} \label{omega-integral}
			0 - \mu \int_{\delta}^{\ep} \omega^2 rdr - (\ep-\delta) 
			= \lambda \frac{\ep^2-\delta^2}{2}.
		\end{equation}
		Now, differentiating the differential equation in 
		\eqref{omega-Neumann} by $r$ we have
		$$
			\omega'' + \omega'\left( \frac{1}{r} - 2\omega \right) 
			+ \frac{1}{r^2}\left( \frac{1}{\mu} - \omega \right) = 0,
		$$
		and therefore an extremal point $r_0 \in (\delta,\ep)$ for $\omega$
		cannot be a minimum if $\omega(r_0)<0$ 
		and cannot be a maximum if $\omega(r_0)>1/\mu$,
		hence $0 \le \omega \le 1/\mu$.
		It then follows by \eqref{omega-integral} that
		$\lambda \ge -\frac{1}{\mu} - \frac{2}{\ep + \delta}$,
		which proves the lemma.
	\end{proof}
	
	Now, 
	combining the integral term in
	\eqref{stability_qf} with \eqref{bound_close_to_nuclei} 
	using half of the kinetic energy,
	produces the lower bound
	\begin{multline*}
		\langle \psi,H(N,\sR)\psi \rangle \ge \\
		\frac{1}{2m} \int_{\R^{2N}} \sum_j \left( 
			\frac{1}{2} |D_j \psi|^2 - 2m(2Z+1) \bigg(
				2(2m(2Z+1) + \ep^{-1}) \chi_{B_\bep}(\bx_j) 
					\right. \\ \left.\left.
				+ \frac{\chi_{{B_\bep}^c}(\bx_j)}{d_{\sR}(\bx_j)}
				\right) |\psi|^2 \right) \,d\sx \\
		\ge \frac{1}{2m} \int_{\R^{2N}} \sum_j \left( 
			\frac{1}{2} |D_j \psi|^2 
				- V_1(\bx_j)
				|\psi|^2 \right) \,d\sx
			\ - (2Z+1)bN,
	\end{multline*}
	with the one-particle potential $V_1(\bx) := $
	$$
		2m(2Z+1) \left(
				2\left( 2m(2Z+1) + \frac{1}{\ep(\bx)} \right) \chi_{B_\bep}(\bx) 
				+ \left( \frac{1}{d_{\sR}(\bx)} - b\right) \chi_{{B_\bep}^c}(\bx)
				\right),
	$$
	where we have here also used a standard method to handle the 
	non-square-integrability of the potential at infinity, 
	by adding and subtracting a constant $b>0$.
	Applying the Lieb-Thirring inequality \eqref{LT-anyons} 
	with the resulting 
	potential $V_1$ given above, we find
	\begin{multline*}
		\langle \psi, H(N,\sR)\psi \rangle \\
		\ge -2mC_{\uA}' \,\nu^2 (2Z+1)^2 \int_{\R^2} \left[
			2\left( 2m(2Z+1) + \ep^{-1} \right) \chi_{B_\bep} 
			+ \left( \frac{1}{d_{\sR}} - b\right) \chi_{{B_\bep}^c}
			\right]_+^2 \\
			- (2Z+1)bN + \frac{Z^2}{8} \sum_{k=1}^K \frac{1}{D_k} 
	\end{multline*}
	\begin{multline*}
		\ge -2mC_{\uA}' \,\nu^2 (2Z+1)^2 \sum_{k=1}^K \bigg(
				\int_{B_{\ep_k}(\bR_k)} 4\left( 2m(2Z+1)+\ep_k^{-1} \right)^2 \,d\bx \\
				+ \int_{{B_{\ep_k}(\bR_k)}^c} \left[ \frac{1}{|\bx - \bR_k|} - b \right]_+^2 d\bx
				- \frac{Z^2/(2Z+1)^2}{16mC_{\uA}' \nu^2 D_k}
			\bigg) \ - (2Z+1)bN 
	\end{multline*}
	\begin{multline*}
		\ge -2mC_{\uA}' \,\nu^2 (2Z+1)^2 \sum_{k=1}^K \bigg(
				4\left( 2m(2Z+1)+\ep_k^{-1} \right)^2 \pi\ep_k^2 \\
				+ 2\pi \int_{\ep_k < r < b^{-1}} (r^{-1}-b)^2 r\,dr
				- \frac{1}{144mC_{\uA}' \nu^2 D_k}
			\bigg) \ - (2Z+1)bN 
	\end{multline*}
	\begin{multline*}
		\ge -4\pi mC_{\uA}' \,\nu^2 (2Z+1)^2 \sum_{k=1}^K \bigg( 
			4 + 16\frac{m^2(2Z+1)^2}{b^2} (\ep_k b)^2 \\
			+ \left[ - \frac{1}{2}(\ep_k b)^2 + 2\ep_k b - \ln{\ep_k b} - \frac{3}{2} \right]_{\ep_k b < 1}
			- \frac{b}{288\pi mC_{\uA}' \nu^2} (D_k b)^{-1}
			\bigg) \\
			- (2Z+1)bN.
	\end{multline*}
	Finally, taking $\ep_k := \min \{b^{-1}, D_k\}$
	and in the latter case using that 
	$$
		-\ln x - \frac{c}{x} \ \le \ [-\ln c - 1]_+
	$$
	for $0<x\le 1$ and $c>0$, 
	we have the further bound
	\begin{multline*}
		\ge -4\pi mC_{\uA}' \,\nu^2 (2Z+1)^2 K \left( 
			5 + 16\frac{m^2(2Z+1)^2}{b^2} 
			+ \left[ \ln \frac{288\pi mC_{\uA}' \nu^2}{b} \right]_+
			\right) \\ 
			- (2Z+1)bN \\
		\ge -C \,\nu^2 m(2Z+1)^2 (K + N),
	\end{multline*}
	for some $C>0$,
	where in the last step we chose $b := \nu^2 m(2Z+1)$.
	\qed

\subsection{External potentials in the Calogero-Sutherland case} \label{sec:Heisenberg_apps}

	Let us end with also considering an application of the above local exclusion 
	and Lieb-Thirring inequalities for one dimensional statistics
	in the Calogero-Sutherland case. We think of $N$ identical such particles
	placed in an external confining potential $V$,
	and would like to obtain a lower bound for the ground state energy $E_0$.
	We have from Theorem \ref{thm:LT-Heisenberg} that for $\alpha \ge 1$
	\begin{multline*}
		E_0 := \inf \left\{ 
			\int_{\R^N} \left( T_{\uH}^N(\psi;\sx) + \sum_{j=1}^N V(x_j) |\psi|^2 \right) d\sx
			\ : \ \psi \in \Dc_{\uH}^\alpha,\ \|\psi\| = 1 \right\} \\
		\ge \inf \left\{ 
			\int_{\R} \left( C_{\uH} \, \xi_{\uH}(\alpha)^2 \tilde{\rho}(x)^3 
				+ V(x) \rho(x) \right) dx
			\ : \ \rho: \R \to \R_{\ge 0}, 
			\int_{\R} \rho = N \right\}.
	\end{multline*}
	Because of the local approximation $\tilde{\rho}$ to the density,
	this minimization problem is less tractable than 
	e.g. the anyonic case in Section \ref{sec:anyon_osc}.
	However, what we can do instead is to consider a partition $\Pc$ of the real
	line into a collection of finite subintervals $\{I_j\}_{j=1}^M$ 
	and an exterior domain $I_{\textup{ext}}$,
	and on each such interval apply Lemma \ref{lem:local_exclusion_H}
	or the local form of Theorem \ref{thm:LT-Heisenberg}:
	\begin{multline*}
		T_{\uH} + \int_{\R} V \rho 
		\ \ge \ \sum_{j=1}^M \left( T_{\uH}^{I_j} + \int_{I_j} V\rho \right) 
			+ \int_{I_{\textup{ext}}} V\rho \\
		\ge \sum_{j=1}^M \left( \frac{\xi_{\uH}(\alpha)^2}{|I_j|^2} \Ec\left( \int_{I_j}\rho \right) 
			+ V_j \int_{I_j}\rho \right) 
			+ V_\textup{ext} \int_{I_\textup{ext}}\rho,
	\end{multline*}
	where we define
	$V_j := \inf_{I_j} V$ and
	$$
		\Ec(\rho_j) := \max \{ 0,\ \rho_j-1,\ \chi_{\{\rho_j \ge 2\}} \,\rho_j^3/32 \},
		\qquad \rho_j := \int_{I_j} \rho.
	$$
	Here we used \eqref{local_exclusion_density_H}
	and \eqref{kinetic-LT-Heisenberg-interval}
	with $C_{\uH} \ge 1/32$, 
	and note that $\rho_j-1 \le \rho_j^3/32$ 
	and $\rho_j \ge 2$ implies 
	$\rho_j \ge \rho_c$, with the critical point
	$\rho_c \approx 5.068 > 2$.
	We have hence reduced to a finite-dimensional minimization problem
	$$
		E_0 \ge E[\Pc] := \inf \left\{ \Ec_\Pc[\brho]
			\ : \ \brho = ((\rho_j),\rho_\textup{ext}) \in \R_{\ge 0}^{M+1},
			\ \sum_j \rho_j + \rho_\textup{ext} = N \right\},
	$$
	$$
		\Ec_\Pc[\brho] :=
			\sum_{j=1}^M \left( \frac{\xi_{\uH}(\alpha)^2}{|I_j|^2} \Ec(\rho_j) 
			+ V_j \rho_j \right) 
			+ V_\textup{ext} \rho_\textup{ext},
	$$
	for which a minimizer $\brho$ exists due to the convexity of $\Ec_\Pc$.
	Depending on $N$ and the details of $V$, 
	a hopefully good lower bound for $E_0$ 
	(and an estimate for the ground state density $\rho$)
	can then be found by maximizing over all suitable partitions $\Pc$,
	$$
		E_0 \ge \sup \{ E[\Pc] : \Pc = ((I_j)_{j=1}^M, I_\textup{ext}) \ \text{partition of $\R$, $M \in \N$} \}.
	$$
	As an illustration of this procedure we can 
	consider the following example, which in the harmonic case 
	$V(x) = \frac{\omega^2}{2}|x|^2$
	can be compared to the exact
	ground state energy for the corresponding Calogero-Sutherland model
	\cite{Calogero:69,Sutherland:71},
	$E_0 = \frac{1}{2}\omega N(1 + \alpha (N-1))$.
	We can also compare to the approximative Thomas-Fermi method discussed in
	\cite{Sen-Bhaduri:95,Smerzi:96}.
	
	\begin{prop}
		Given an external confining potential 
		$V(x) = c^\mu |x|^\mu$, with $c,\mu > 0$,
		and assuming
		$N \gg (\xi_{\uH}(\alpha) c)^{2/\mu}$ and $\xi_{\uH}(\alpha) c N \gg 1$,
		we have for the ground state energy of
		$H = \hat{T}_{\uH} + \sum_{j=1}^N V(x_j)$
		\begin{equation} \label{Heisenberg_energy_bound}
			E_0 \ \ge \ C(\mu,N)
				\left( \xi_{\uH}(\alpha) \,c \right)^{\frac{2\mu}{\mu+2}} 
				N^{\frac{3\mu+2}{\mu+2}},
		\end{equation}
		where $C(\mu,N) \ge 0$ is a constant s.t.
		$$
			\lim_{N \to \infty} C(\mu,N) = \left( \frac{\sqrt{3}}{4\sqrt{2\pi}} \right)^{\frac{2\mu}{\mu+2}}
				\frac{\mu+2}{2\mu^2} \frac{ \Gamma(\frac{1}{\mu}) }{ \Gamma(\frac{5}{2}+\frac{1}{\mu}) }
				\left( \frac{ \Gamma(\frac{3}{2}+\frac{1}{\mu}) }{ \Gamma(1+\frac{1}{\mu}) }  \right)^{\frac{3\mu+2}{\mu+2}}.
		$$
		In particular, 
		$\liminf_{N \to \infty} E_0/N^2 \ge \frac{\sqrt{3}}{8\pi} \xi_{\uH}(\alpha) \omega$ 
		for $\mu=2$ and $c=\omega/\sqrt{2}$.
	\end{prop}
	\begin{proof}
		Consider a simple partition of the form
		$$
			\Pc = \left( [-(k+1)a,-ka]_{k=0}^{M-1}, [ka,(k+1)a]_{k=0}^{M-1}, (-\infty,-Ma] \cup [Ma, +\infty) \right),
		$$
		$a>0$, for which we have 
		(making use of the symmetry, and with a suitable relabeling of $\rho_j$)
		$$
			\Ec_\Pc[\brho] =
			2\sum_{k=0}^{M-1} \left(
				\frac{\xi_{\uH}(\alpha)^2}{a^2} \Ec(\rho_k)
				+ c^\mu (ka)^\mu \rho_k
				\right) + c^\mu (Ma)^\mu \rho_\textup{ext}.
		$$
		We extremize the functional
		\begin{multline*}
			F[\brho,\lambda] :=
			\Ec_\Pc[\brho] - \lambda \left(\sum_j \rho_j + \rho_\textup{ext}\right) 
				+ \lambda N \\
			= 2\sum_{k=0}^{M-1} \left(
				\frac{\xi_{\uH}(\alpha)^2}{a^2} \Ec(\rho_k)
				+ \left( c^\mu a^\mu k^\mu - \lambda \right) \rho_k
				\right) 
				+ \left( c^\mu a^\mu M^\mu - \lambda \right) \rho_\textup{ext}
				+ \lambda N
		\end{multline*}
		to find, if $M$ is chosen large enough,
		$\rho_\textup{ext} = 0$ and
		$\rho_k = 0$ for $c^\mu a^\mu k^\mu - \lambda \ge 0$,
		$\rho_k = 1$ for $0 < \lambda - c^\mu a^\mu k^\mu < \xi_{\uH}(\alpha)^2/a^2$,
		and otherwise
		$$
			\rho_k = \max \left\{ \rho_c, \frac{a\sqrt{32/3}}{\xi_{\uH}(\alpha)} \sqrt{\lambda - c^\mu a^\mu k^\mu} \right\}.
		$$
		Hence, $N = 2\sum_k \rho_k + \rho_\textup{ext} =$
		\begin{multline*}
			= \sum_{\substack{k \ge 0\ \text{s.t.}\\ \lambda - c^\mu a^\mu k^\mu \ge \xi_{\uH}(\alpha)^2/a^2}} 
				2\max \left\{ \rho_c, \frac{a\sqrt{32/3}}{\xi_{\uH}(\alpha)} \sqrt{\lambda - c^\mu a^\mu k^\mu} \right\}
				\\ 
			+ \sum_{\substack{k \ge 0\ \text{s.t.}\\ 0<\lambda - c^\mu a^\mu k^\mu < \xi_{\uH}(\alpha)^2/a^2}} 2.
		\end{multline*}
		Now, let $a := 1/c$ and assume that 
		$N \gg (\xi c)^{2/\mu}$ and $\xi c N \gg 1$.
		Then
		\begin{multline*}
			\xi c N 
			\sim \sum_{\substack{k \ge 0\ \text{s.t.}\\ \lambda - k^\mu \ge \xi^2 c^2}} 
				2\sqrt{32/3} \sqrt{\lambda - k^\mu} \\
			\sim 8\sqrt{2/3} \int_0^{\lambda^{1/\mu}} \sqrt{\lambda-k^\mu} \,dk 
			= 8\sqrt{2/3} \,I(\mu) \lambda^{1/2+1/\mu},
		\end{multline*}
		where
		$$
			I(\mu) := \int_0^1 \sqrt{1-x^\mu} \,dx 
			= \frac{\sqrt{\pi}}{2} \frac{ \Gamma(1+\frac{1}{\mu}) }{ \Gamma(\frac{3}{2}+\frac{1}{\mu}) }.
		$$
		Hence, 
		$\lambda \sim \left( \frac{\sqrt{3/2} }{ 8 I(\mu) } \xi c N \right)^{\frac{2\mu}{\mu+2}}$ 
		and 
		\begin{multline*}
			E_0 \ge \Ec_\Pc[\brho] \\
			\ge 2\sum_{\substack{k \ge 0\ \text{s.t.}\\ \lambda - k^\mu \ge \xi^2 c^2}} 
				\left( \frac{\xi^2 c^2}{32} \left( \frac{\sqrt{32/3}}{\xi c} \sqrt{\lambda - k^\mu} \right)^3
				+ k^\mu \left( \frac{\sqrt{32/3}}{\xi c} \sqrt{\lambda - k^\mu} \right) 
				\right) \\
			= \frac{8\sqrt{2/3}}{\xi c} \sum_k \left(
				\frac{1}{3} (\lambda - k^\mu)^{\frac{3}{2}}
				+ k^\mu \sqrt{\lambda - k^\mu} \right) 
			\ \sim 8\sqrt{2/3} \,J(\mu) \frac{ \lambda^{\frac{3}{2}+\frac{1}{\mu}} }{ \xi c },
		\end{multline*}
		where
		$$
			J(\mu) := \frac{1}{3} \int_0^1 (1-x^\mu)^{\frac{3}{2}} \,dx 
			+ \int_0^1 x^\mu \sqrt{1-x^\mu} \,dx 
			= \frac{\sqrt{\pi}}{2} \frac{\mu+2}{2\mu^2} \frac{ \Gamma(\frac{1}{\mu}) }{ \Gamma(\frac{5}{2}+\frac{1}{\mu}) }.
		$$
		It follows under the above conditions that
		$$
			\liminf_{N \to \infty} \frac{E_0}{N^\frac{3\mu+2}{\mu+2}} 
			\ge \left( 8\sqrt{2/3} \right)^{-\frac{2\mu}{\mu+2}} \frac{J(\mu)}{I(\mu)^{\frac{3\mu+2}{\mu+2}}} 
				\left( \xi c \right)^{\frac{3\mu+2}{\mu+2} - 1},
		$$
		which is the asymptotics of the r.h.s. of 
		\eqref{Heisenberg_energy_bound}.
	\end{proof}

\end{document}